\pdfoutput=1
\documentclass[cmp]{svjour}
\usepackage{amsmath,amsfonts,amssymb,color,epsfig,graphicx} 

\usepackage[percent]{overpic}
\usepackage{subfigure} 
\usepackage{pifont} 
\usepackage{cite} 
\newcommand{\yestick}{\ding{51}}
\newcommand{\notick}{\ding{55}}
\usepackage{multirow}

\def\pro{\mathop{\rm PRO}}
\def\sep{\mathop{\rm SEP}}
\def\rank{\mathop{\rm rank}}
\def\tr{\mathop{\rm Tr}}
\def\iso{\mathop{\rm iso}}

\usepackage{bbm}
\newcommand{\one}{\mathbbm{1}} 
\newcommand{\I}{\text{i}} 
\newcommand{\E}{\text{e}} 
\newcommand{\etal}{\textit{et al. }}
\newcommand{\Lam}{{\Lambda^{2}}}
\newcommand{\Lamf}{{\Lambda^{2}_{\text{f}}}}
\newcommand{\Lamm}{{\Lambda^{2}_{\text{m}}}}
\newcommand{\Lamx}{{\Lambda^{2}_{\text{x}}}}
\newcommand{\Ee}{{\text{E}}}
\newcommand{\G}{{\text{G}}}
\newcommand{\Gl}{{\text{G}}_{\text{l}}}
\newcommand{\Gc}{{\text{G}}^{\text{c}}}
\newcommand{\Gcl}{{\text{G}}^{\text{c}}_{\text{l}}}
\newcommand{\Gm}{{\text{G}}^{\text{m}}}
\newcommand{\Gml}{{\text{G}}^{\text{m}}_{\text{l}}}
\newcommand{\Gf}{{\text{G}}^{\text{f}}}
\newcommand{\Gfl}{{\text{G}}^{\text{f}}_{\text{l}}}
\newcommand{\Gfc}{{\text{G}}^{\text{f/c}}}
\newcommand{\Gt}{{\text{G}}^{{\text{t}}}}
\newcommand{\Gx}{{\text{G}}^{\text{x}}}
\newcommand{\Gxl}{{\text{G}}^{\text{x}}_{\text{l}}}
\newcommand{\ET}{{\text{E}}_{\text{T}}}
\newcommand{\EB}{{\text{E}}_{\text{B}}}
\newcommand{\ER}{{\text{E}}_{\text{R}}}

\newcommand{\DT}{D_{\text{T}}}
\newcommand{\DB}{D_{\text{B}}}
\newcommand{\ETw}{\widetilde{{\text{E}}_{\text{T}}}}
\newcommand{\EBw}{\widetilde{{\text{E}}_{\text{B}}}}
\newcommand{\DTw}{\widetilde{D_{\text{T}}}}
\newcommand{\DBw}{\widetilde{D_{\text{B}}}}
\newcommand{\Fw}{\widetilde{F}}

\newcommand{\VN}{{\text{S}}}
\newcommand{\LI}{{\text{M}}}
\newcommand{\channel}{{\mathcal{E}}}
\newcommand{\SH}{{\mathcal{S}({\mathcal{H}})}}
\newcommand{\co}{\text{c}_{\vartheta}}
\newcommand{\si}{\text{s}_{\vartheta}}

\newcommand{\cH}{{\mathcal{H}}}
\newcommand{\cC}{{\mathcal{C}}}
\newcommand{\cS}{{\mathcal{S}}}
\newcommand{\bC}{{\mathbb{C}}}

\newcommand{\bRp}{\mathbb{R}^{+}}
\newcommand{\bCd}{\mathbb{C}^{d}}

\newcommand{\bra}[1]{\langle#1|}
\newcommand{\ket}[1]{|#1\rangle}
\newcommand{\proj}[1]{| #1\rangle\!\langle #1 |}
\newcommand{\ketbra}[2]{|#1\rangle\!\langle#2|}
\newcommand{\braket}[2]{\langle#1|#2\rangle}

\newcommand{\abs}[1]{|#1|}

\newcommand{\phim}{\ket{\phi_{\text{m}}}}

\begin{document}

\title{A comparison of old and new definitions of the geometric
  measure of entanglement}

\author{Lin Chen\inst{1} \and Martin Aulbach\inst{2} \and Michal Hajdu\v{s}ek\inst{3}
}

\institute{Department of Pure Mathematics and Institute for Quantum
  Computing, University of Waterloo, Waterloo, Ontario, N2L 3G1,
  Canada, \email{linchen0529@gmail.com} \and Centre for Quantum
  Technologies, National University of Singapore, 3 Science Drive 2,
  Singapore 117542, \email{aulbach.martin@gmail.com} \and Department
  of Physics, Graduate School of Science, University of Tokyo, 7-3-1
  Hongo, Bunkyo-ku, Tokyo, Japan, 113-0033}

\date{}
\maketitle

\begin{abstract}
  Several inequivalent definitions of the geometric measure of
  entanglement (GM) have been introduced and studied in the past.
  Here we review several known and new definitions, with the
  qualifying criterion being that for pure states the measure is a
  linear or logarithmic function of the maximal fidelity with product
  states.  The entanglement axioms and properties of the measures are
  studied, and qualitative and quantitative comparisons are made
  between all definitions.  Streltsov \etal [New J.~Phys \textbf{12}
  123004 (2010)] proved the equivalence of two linear definitions of
  GM, whereas we show that the corresponding logarithmic definitions
  are distinct.  Certain classes of states such as ``maximally
  correlated states'' and isotropic states are particularly valuable
  for this analysis.  A little-known GM definition is found to be the
  first one to be both normalized and weakly monotonous, thus being a
  prime candidate for future studies of multipartite entanglement.  We
  also find that a large class of graph states, which includes all
  cluster states, have a ``universal'' closest separable state that
  minimizes the quantum relative entropy, the Bures distance and the
  trace distance.
\end{abstract}

\section{Introduction}\label{sec:introduction}

Entanglement measures lie at the heart of quantum information theory,
because they assess the usefulness of quantum states for tasks such as
quantum teleportation, quantum computation and cryptography protocols
\cite{pv07,hhh09}.  Numerous entanglement measures have been defined
in the past, each of which may capture different properties of a state
as a resource for certain tasks.  One well-known measure is the
geometric measure of entanglement (GM).  Originally introduced for
pure bipartite states \cite{shimony95,bh01}, the GM was subsequently
generalized to multipartite and to mixed states
\cite{bl01,wg03,ssb06}.  Two key benefits of GM are that it is an
inherently multipartite entanglement measure and that it is
comparatively easy to compute for many states
\cite{wg03,zch10,skb11PRA}.

The GM has a variety of operational interpretations: It assesses the
usability of initial states for Grover's algorithm \cite{bno02,ssb04},
the discrimination of quantum states under LOCC \cite{mmv07}, the
additivity and output purity of quantum channels \cite{wh02} and the
usefulness of states as resources for one-way quantum computation
\cite{zch10,gfe09,ndm07,mpm10}. Further uses of GM include the
construction and study of entanglement witnesses \cite{wg03,hmm08},
the derivation of a generalized Schmidt decomposition \cite{chs00} and
the study of condensed matter systems, such as ground state
characterization and detection of phase transitions
\cite{odv08,orus08,ow10}.

Several inequivalent definitions of GM have surfaced in the literature
\cite{bl01,wg03,ssb06,weg04,bno02,mmv07,hmm06,skb11PRA}.
Regarding pure states, GM is expressed either as a linear or
logarithmic function of the maximal fidelity with product
states. Regarding mixed states, the pioneering papers did not agree on
a unique definition,
which led to the emergence of several inequivalent extensions of GM to
mixed states.
Although some of the GM definitions have been compared to other
entanglement measures \cite{weg04}, a detailed comparison of all the
different definitions of GM to each other has not been done before. An
important milestone towards this goal was achieved by Streltsov
\emph{et al.}~\cite{skb11NJP} who proved the equivalency of two
frequently used definitions of GM.

In this paper we compare and characterize several known and new
definitions of GM. The only qualifying criterion for an entanglement
quantity to be regarded a GM definition is that on the subset of pure
states it coincides with the well-defined linear or logarithmic GM.

Five known definitions, one little-known and one new definition of GM
are studied in this paper, first with respect to their entanglement
axioms.  This is followed by a quantitative and qualitative comparison
of the definitions to each other.  The ``maximally correlated states''
(as defined in Definition \ref{defmaxcorr}) turn out to be a
particularly helpful class of states for this purpose.  We also
discover that a large class of graph states, including all cluster
states, have a ``universal'' closest separable state that minimizes
the quantum relative entropy, the Bures distance and the trace
distance.

The paper is organized as follows.  Section \ref{sec:prelim} reviews
some basic concepts of quantum information theory for later usage.  In
Section \ref{sec:definitions} the definitions of GM are introduced,
and some preliminary results, e.g.~with regard to entanglement axioms,
are obtained. The subsequent Section \ref{sec:relationships} closely
examines the relationship between all six distinct definitions of GM
from a variety of perspectives, and a hierarchy that allows for a
partitioning of state space is obtained.  A common closest separable
state with respect to three different distance measures is presented
for a large class of graph states in Section \ref{graphstates}.  The
concluding Section \ref{sec:con} summarizes our results.  For
convenience, Table \ref{tab:GMinequalities}, Table \ref{GMproperties},
Figure \ref{GMhierarchy} and Figure \ref{statespace} list some of our
findings in compact form.

\section{Preliminaries}\label{sec:prelim}

First, we review some basic concepts for later usage, in particular
axiomatic entanglement measures and distance measures.  For more
comprehensive reviews we refer to \cite{pv07,hhh09,nc00book,bz06book}.

\subsection{Axiomatic entanglement measures}\label{axioms}

Operationally motivated entanglement measures such as the Entanglement
Cost and the Distillable Entanglement have clear physical meanings,
but tend to be difficult to study from a mathematical viewpoint,
especially for multipartite systems.  On the other hand, axiomatically
motivated entanglement measures may not have operational implications.

Considering $n$ parties $A_1, \ldots, A_n$ with joint Hilbert space
$\cH = \otimes^{n}_{j=1} \cH_j$, a general $n$-partite state shared
over the parties is described by a density matrix $\rho \in \SH$
acting on $\cH$.  Such a state is considered separable if it can be
written in the form $\rho = \sum_i p_i \rho^{1}_{i} \otimes \cdots
\otimes \rho^{n}_{i}$, with $\sum_i p_i = 1$, and where $\rho^{j}_{i}$
is a single-particle state of the $j$-th party.  In the axiomatic
approach, an entanglement measure is a functional $\Ee : \SH \to \bRp$
that satisfies two fundamental axioms:
\begin{enumerate}
\item $\Ee ( \rho ) = 0 \:$ if $\rho$ is separable.
  
\item $\Ee$ does not increase on average under local operations and
  classical communication (LOCC). Depending on which quantum
  operations are considered, this is defined as
  \begin{enumerate}
  \item weak monotonicity: \quad $\Ee ( \rho ) \geq \Ee ( \sigma ) \,
    ,$ if $\rho \stackrel{\text{LOCC}}{\longmapsto} \sigma = \channel
    ( \rho ) = \sum_i \vec{P}_i \rho \vec{P}_i^{\dagger} \, .$

  \item strong monotonicity: $\Ee ( \rho ) \geq \sum_{i} p_{i} \Ee (
    \sigma_{i} ) \, ,$ if $\rho \stackrel{\text{LOCC}}{\longmapsto}
    \sigma_{i} = \frac{\channel_i (\rho)} {\tr \channel_i (\rho)}$
    with probability $p_i = \tr \channel_i (\rho) \, .$
  \end{enumerate}
\end{enumerate}
Here, the maps $\channel$ and $\channel_i$ stand for LOCC, and the
elements $\{ \vec{P}_i \}$ form a complete POVM, i.e.~$\sum_i
\vec{P}_i^{\dagger} \vec{P}_i = \one$.  Weak monotonicity corresponds
to trace-preserving quantum operations where the measurement outcome
is unknown or discarded. Strong monotonicity corresponds to selective
quantum operations, also known as measuring quantum operations
\cite{pv07}. Weak monotonicity is a special case of strong
monotonicity that follows when a single outcome is obtained with
probability $1$, i.e.~$\channel = \channel_1$.  We refer to measures
satisfying Axioms 1 and 2(a) as \emph{weak entanglement measures}, and
measures that additionally satisfy Axiom 2(b) as \emph{strong
  entanglement measures}.

Historically, strong monotonicity was required for entanglement
measures \cite{vidal00,vp98,pv07}, but there is now a consensus that
weak monotonicity suffices \cite{hhh09}. Weak entanglement measures
can thus be considered proper entanglement measures.  Another historic
requirement is that entanglement measures should coincide with the
entropy of entanglement for pure bipartite states
\cite{vp98,pv07}. However, many popular measures fail this property,
and the property cannot be easily extended to multipartite states, so
it is not considered essential anymore. Invariance under local unitary
(LU) operations is clearly important for entanglement measures, but it
does not need to be stated as a separate axiom, because it
automatically follows from weak monotonicity \cite{pv07}.

Apart from the axioms discussed above, many more desirable properties
could be specified.  Some common ones are the following\footnote{For
  brevity, we abbreviate $f ( \proj{\psi} )$ as $f ( \ket{\psi} )$ for
  functions $f ( \rho )$ defined on $\SH$.}:
\begin{itemize}
\item Normalization: \quad $\Ee ( \ket{\Phi}^{\otimes n} ) = n \, ,$
  \quad for 2 qubit Bell states $\ket{\Phi}$
  
\item Convexity: \quad $\Ee ( \rho ) \leq \sum_{i} p_{i} \Ee (
  \rho_{i} ) \, ,$ \quad for all $\rho = \sum_{i} p_{i} \rho_{i}$
  
\item Additivity: \quad $\Ee ( \rho^{\otimes 2} ) = 2 \, \Ee ( \rho )
  \, ,$ \quad for all $\rho \in \SH$
  
\item Strong Additivity: \quad $\Ee ( \rho \otimes \sigma ) = \Ee (
  \rho ) + \Ee ( \sigma ) \, ,$ \quad for all $\rho , \sigma \in \SH$
\end{itemize}

The desirability of normalization is clear from the perception that
Bell states carry 1 ebit of entanglement each.  Convexity is motivated
by the notion that entanglement should not increase under loss of
information, namely when a selection of identifiable states $\rho_i$
(r.h.s.) is transformed into a mixture $\rho$ (l.h.s.) \cite{pv07}.
One may assume that this process can be physically realized by
standard quantum operations, and thus strong monotonicity implies
convexity.  However, some additional properties (such as continuity)
need to be satisfied, and the logarithmic negativity constitutes a
counterexample by being a strong entanglement measure that is not
convex \cite{plenio05}.  An important consequence of convexity is that
the measure can be maximized on the subset of pure states, i.e.~there
exist maximally entangled states (MES) that can be cast as pure states
$\rho = \proj{\psi}$.

If $f$ and $g$ are two convex functions and $g$ is non-decreasing,
then $g \circ f$ is also convex. For example, if $f ( \rho ) \geq 0$
is a convex measure, then $f^2 ( \rho )$ is also convex, using $g(x) =
x^2$.  In analogy to convexity, \emph{concavity} is defined as $f (
\rho ) \geq \sum_{i} p_{i} f( \rho_{i} )$ for all $\rho = \sum_{i}
p_{i} \rho_{i}$.

Regarding the additivity axioms, the tensor products in their
definition have a specific physical meaning: Instead of enlarging the
Hilbert space, $\rho \otimes \sigma$ refers to two states acting on
the same Hilbert space \cite{zch10}. If $\rho$ and $\sigma$ are both
states of $n$ $d$-level subsystems, then $\rho \otimes \sigma$ is a
state of $n$ $d^2$-level subsystems (instead of $2n$ $d$-level
subsystems).  Obviously, strong additivity implies additivity.

From a mathematical viewpoint, two entanglement measures $\Ee_1$ and
$\Ee_2$ are equivalent, if $\Ee_1 (\rho) = \Ee_2 (\rho)$ holds for all
$\rho \in \SH$.  A less restrictive, yet physically sound criterion is
the property of \emph{ordering}: $\Ee_1$ and $\Ee_2$ have the same
entanglement ordering, if the same order is obtained when sorting all
states by their amount of entanglement.  This is the case if for all
$\rho , \sigma \in \SH$ the two statements $\Ee_1 ( \rho) > \Ee_1 (
\sigma )$ and $\Ee_2 ( \rho ) > \Ee_2 ( \sigma )$ are equivalent,
i.e.~they are either both true or both false. Entanglement measures
that are equivalent will be denoted as $\Ee_1 \equiv \Ee_2$, and
measures with the same ordering as $\Ee_1 \cong \Ee_2$.

Many different entanglement measures have been proposed, but here we
only consider measures that are based on the distance to the set of
separable states.  The \emph{relative entropy of entanglement} (REE)
measures the minimum distance in terms of relative entropy between the
given state $\rho$ and the set of separable states ($\sep$):
\begin{equation}\label{eq:REEdef}
  \ER ( \rho ) := \min_{\sigma \in \sep} S (\rho \vert \sigma) \, ,
\end{equation}
where
\begin{equation}\label{eq:QREdef}
  S( \rho \vert \sigma ) = \tr \rho( \log \rho - \log \sigma )
\end{equation}
is the \emph{quantum relative entropy} \cite{vp98}.  Any state
$\sigma$ minimizing $S(\rho \vert \sigma)$ is called a \emph{closest
  separable state} of $\rho$.  Since the definition involves the
minimization over all separable states, REE is known only for a few
examples, such as bipartite pure states \cite{vpr97,vp98,pv01}, Bell
diagonal states \cite{vpr97,rains99,rains99erratum}, Werner states
\cite{aej01,vw01,rains01}, maximally correlated states, isotropic
states \cite{rains99,rains99erratum}, generalized Dicke states
\cite{weg04,wei08,hmm08}, antisymmetric basis states
\cite{weg04,hmm08}, some graph states \cite{mmv07,hm13}, the Smolin
state, and D\"ur's multipartite entangled states \cite{wag04,wei08}. A
numeric method for computing REE of bipartite states has been proposed
\cite{vp98}.

The REE can be applied to arbitrary multipartite states, pure or
mixed. The central measure of this paper, the geometric measure of
entanglement (GM) -- to be discussed in Section \ref{sec:definitions}
-- is also an inherently multipartite measure, although its definition
for mixed states is not unique.

\subsection{Distance measures and fidelity}\label{fid}

A good measure of distance $D ( \rho , \sigma ) : \SH \times \SH \to
\bRp$ between two quantum states should be symmetric, zero iff $\rho =
\sigma$, and observe weak monotonicity, which in this context means
$D( \rho , \sigma ) \geq D ( \channel (\rho ) , \channel ( \sigma ) )$
for any trace-preserving quantum operation $\channel$
\cite{vpr97,hhh09,nc00book}.  The last property, also known as
\emph{contractivity} under quantum operations, guarantees LU
invariance: $D ( \rho , \sigma ) = D ( U \rho U^{\dagger} , U \sigma
U^{\dagger} )$.  Any distance function with these properties is called
a \emph{distance measure}.  One such distance measure is the trace
distance \cite{nc00book},
\begin{equation}\label{trdist}
  \DT (\rho , \sigma ) =
  \frac{1}{2} \tr \abs{\rho - \sigma } =
  \frac{1}{2} \tr \sqrt{( \rho - \sigma )^2} =
  \frac{1}{2} \sum_i \abs{\lambda_i} \, ,
\end{equation}
where the $\lambda_i$ are the eigenvalues of the matrix $(\rho -
\sigma )$.  For qubits $\DT (\rho , \sigma )$ is equal to half the
Euclidean distance between the corresponding Bloch vectors.  The trace
distance is convex in both arguments, and furthermore is jointly
convex \cite{nc00book}:
\begin{equation}\label{trjointconvex}
  \DT \Big( \sum_i p_i \rho_i , \sum_i p_i \sigma_i \Big)
  \leq \sum_i p_i \DT ( \rho_i , \sigma_i ) \, .
\end{equation}
Another distance measure is the Bures distance and the closely related
fidelity \cite{uhlmann76,jozsa94,nc00book}.  The Bures distance is
\begin{equation}\label{bures}
  \DB ( \rho , \sigma ) = \sqrt{ 2 - 2 F( \rho , \sigma )} \, ,
\end{equation}
where $F (\rho , \sigma )$ is the fidelity between two states, defined
as
\begin{equation}\label{fidelity}
  F ( \rho , \sigma ) = \tr \sqrt{\sqrt{\rho} \sigma
    \sqrt{\rho}} \, .
\end{equation}
Alternative definitions in the literature are $\DTw := \DT^2$ for the
trace distance, $\DBw := \DB^2$ for the Bures distance, and $\Fw :=
F^2$ for the fidelity\footnote{This ambiguity of the definitions led
  to an incorrect definition of the Bures distance in
  \protect\cite{cw07}.}. The three necessary properties of distance
measures outlined above remain invariant under exponentiation, and
therefore $\DTw$ and $\DBw$ are also distance measures.  Although
closely related to the Bures distance, the fidelity is not a distance
measure itself.  The fidelity is symmetric, unitarily invariant,
concave in both arguments, jointly concave, and has codomain $F \in
[0,1]$, with $F=1$ iff $\rho = \sigma$ \cite{nc00book}.  According to
Uhlmann's theorem, the fidelity has a clear physical interpretation as
the maximal overlap between all purifications of the input states
\cite{uhlmann76,nc00book}.

If at least one of the two arguments of the fidelity is pure, then
\eqref{fidelity} simplifies to $F^2 ( \rho , \sigma ) = \tr ( \rho
\sigma )$, thus yielding
\begin{subequations}\label{purefidelity}
  \begin{align}
    F^2 ( \ket{\psi} , \sigma )& =
    \abs{\bra{\psi} \sigma \ket{\psi}}  \, , \label{pf1} \\
    F^2 ( \ket{\psi} , \ket{\phi} )& =
    \abs{\braket{\psi}{\phi}}^2 \, . \label{pf2}
  \end{align}
\end{subequations}
In particular, for pure states the fidelity coincides with the
Fubini-Study metric, the natural geometry on $\cH$.

The fidelity also provides upper and lower bounds on the trace
distance, with the lower bound increasing with the purity of the input
states \cite{nc00book}:
\begin{subequations}\label{trfidineq}
  \begin{align}
    1- F ( \rho , \sigma ) \leq \DT ( \rho , \sigma )&
    \leq \sqrt{1- F^2 ( \rho , \sigma )} \, , \label{tf1} \\
    1- F^2 ( \ket{\psi} , \sigma ) \leq \DT ( \ket{\psi} , \sigma )&
    \leq \sqrt{1- F^2 ( \ket{\psi} , \sigma )} \, , \label{tf2} \\
    \DT ( \ket{\psi} , \ket{\phi} )&
    = \sqrt{1- F^2 ( \ket{\psi} , \ket{\phi} )} \, . \label{tf3}
  \end{align}
\end{subequations}
Vedral \etal \cite{vpr97} found that from every distance measure $D (
\rho , \sigma )$ a weak entanglement measure $\Ee (\rho )$ can be
constructed as
\begin{equation}\label{vdef}
  \Ee ( \rho ) := \min_{\sigma \in \sep} D ( \rho , \sigma ) \, .
\end{equation}
This construction directly yields weak entanglement measures from the
trace distance \eqref{trdist} and Bures distance \eqref{bures},
\begin{subequations}\label{traceburesentanglement}
  \begin{align}
    \ET ( \rho )& = \min_{\sigma \in \sep}
    \DT ( \rho , \sigma ) \, , \label{tr1} \\
    \EB ( \rho )& = \min_{\sigma \in \sep}
    \DB ( \rho , \sigma ) \, , \label{tr2}
  \end{align}
\end{subequations}
which we refer to as trace entanglement (TE) and Bures entanglement
(BE), respectively.  In analogy to $\DTw$ and $\DBw$, we define $\ETw
:= \ET^2$ and $\EBw := \EB^2$, and note that they are also weak
entanglement measures.  From \eqref{bures} and the l.h.s.~of
\eqref{tf1}, it follows that
\begin{equation*}
  \EBw ( \rho ) \leq
  2 \ET ( \rho ) \, , \quad \text{for all } \rho \in \SH \, .
\end{equation*}

The quantum relative entropy \eqref{eq:QREdef} is not symmetric, and
therefore not a proper distance measure. Nevertheless, by means of
\eqref{vdef}, it gives rise to the REE, which is a strong entanglement
measure \cite{vp98,horodecki05}.  One could also ask whether the
Hilbert-Schmidt distance $D_{\text{HS}} ( \rho , \sigma ) = \tr [ (
\rho - \sigma )^2 ]$, a metric in the mathematical sense, gives rise
to an entanglement measure.  However, this metric does not satisfy
weak monotonicity, and it is an open question whether inserting
$D_{\text{HS}}$ in \eqref{vdef} yields an entanglement measure
\cite{ozawa00}.

A simple, but important mathematical inequality for this paper is
\begin{equation}\label{elin}
  x - 1 \geq \log_b x \qquad \forall \, x \in (0,1] \:\:
  \forall \, 1 < b \leq \text{e} \, ,
\end{equation}
where $\E$ denotes the base of the natural logarithm.  We call
\eqref{elin} the \emph{elementary inequality}.  To demonstrate its
usefulness, consider the two most common entropic quantities in
quantum information theory, the \emph{linear entropy} $\LI (\rho ) = 1
- \tr (\rho^2 )$, and the \emph{von Neumann entropy} $\VN (\rho ) = -
\tr (\rho \log \rho )$.  The linear entropy can be understood as an
approximation of von Neumann entropy, obtained by the Taylor series
$\log (\rho ) \approx \rho - \one$, where $\one$ has the same range as
$\rho$.  For the commonly used logarithm bases $2$ and $\E$,
\eqref{elin} yields $\one - \rho \leq -\log \rho$, hence
\begin{equation}
  \LI ( \rho ) \leq \VN ( \rho ) \, , \quad \forall \, \rho \in \SH \, .
\end{equation}

\section{Geometric measure of entanglement}\label{sec:definitions}

In this section we review the two two common definitions of GM for
pure states, and introduce the known and unknown extensions to mixed
states.

\subsection{GM for pure states}\label{sec:pure}

The fundamental quantity for GM of pure states is
\begin{equation}\label{purelambda}
  \Lam (\ket{\psi}) := \max_{\ket{\varphi} \in \pro}
  \abs{\braket{\varphi}{\psi}}^2 \, ,
\end{equation}
where $\pro$ denotes the set of fully product pure states of $\cH$,
henceforth referred to as product states.  Comparing
\eqref{purelambda} with \eqref{purefidelity}, we see that $\Lam
(\ket{\psi})$ is the maximum fidelity between $\ket{\psi}$ and the set
$\pro$.  Furthermore, it is clear from \eqref{purefidelity} that the
maximal value of $F ( \ket{\psi} , \cdot )$ can be found among pure
states:
\begin{equation}\label{purelambda2}
  \Lam (\ket{\psi}) = \max_{\ket{\varphi} \in \pro}
  F^2 ( \ket{\psi} , \ket{\varphi} ) = \max_{ \sigma \in \sep}
  F^2 ( \ket{\psi} , \sigma ) \, .
\end{equation}
Any product state or separable state that maximizes the corresponding
fidelity expression in \eqref{purelambda2} is called \emph{closest
  product state} (CPS) or \emph{closest separable state} (CSS),
respectively. Note that the CPS or CSS is in general not unique. The
relationship between the CSSs and CPSs of a given state $\ket{\psi}
\in \cH$ is seen from \eqref{purefidelity}: If $\{ \ket{\phi_i} \}$ is
the set of CPSs, then any superposition $\sigma = \sum_i p_i
\proj{\phi_i}$ with $\sum_i p_i = 1$ is a CSS.  Conversely, if
$\sigma$ is a CSS, then there must exist a decomposition $\sigma =
\sum_i p_i \proj{\phi_i}$ such that all $\ket{\phi_i}$ are CPSs.

The two common definitions of GM for pure states are
\begin{subequations}\label{GMpure}
  \begin{align}
    \G (\ket{\psi} )&  := 1 - \Lam (\ket{\psi}) \, ,
    \label{GMpure1} \\
    \Gl (\ket{\psi})& := - \log \Lam (\ket{\psi}) \, ,
    \label{GMpure2}
  \end{align}
\end{subequations}
which we refer to as the \emph{linear GM} and \emph{logarithmic GM},
respectively. Unless denoted otherwise, the base 2 logarithm is used
in this paper.  Thanks to the elementary inequality \eqref{elin}, the
results of this paper are also valid for any other logarithm base up
to, and including, the base of the natural logarithm. For larger bases
this is not the case, because \eqref{elin} then no longer holds.  Both
$\G$ and $\Gl$ increase monotonically with $\Lam$, and eliminating
$\Lam$ yields
\begin{equation}\label{GGlrelation}
  \Gl ( \ket{\psi} ) = - \log ( 1 - \G ( \ket{\psi} )) \qquad
  \forall \, \ket{\psi} \in \cH \, .
\end{equation}
Due to this monotonic relationship, $\G$ and $\Gl$ have the same
ordering for pure states.  In particular, they have the same MESs. For
bipartite states the MES (up to LU) is $\ket{\Psi} = \frac{1}{\sqrt d}
\sum^d_{i=1} \ket{ii}$, yielding $\G (\ket{\Psi}) = 1 - \frac{1}{d}$
and $\Gl (\ket{\Psi}) = \log d$.  For the simplest multipartite case
of three qubits, the W state has been analytically determined as the
MES for the GM \cite{cxz10}.  For general multipartite systems,
however, W states only yield low entanglement in terms of GM, and the
identification of the MESs is an open problem.  For the subset of
permutation-symmetric states the MES are better understood
\cite{amm10,aulbach12}, because the CPSs of symmetric states are
symmetric themselves \cite{hkw09}, thus considerably simplifying the
optimization problem.

From \eqref{GGlrelation} and the elementary inequality \eqref{elin} it
follows that
\begin{equation}\label{ea:Gg<=Gl}
  \G ( \ket{\psi} ) \leq \Gl ( \ket{\psi} ) \qquad
  \forall \, \ket{\psi} \in \cH \, .
\end{equation}

Since $\Lam (\ket{\psi}) = \abs{\braket{\psi}{\psi}}^2 = 1$ holds for
all $\ket{\psi} \in \pro$, entanglement Axiom 1 is satisfied for pure
states for both definitions in \eqref{GMpure}.  Regarding Axiom 2, an
extension of $\G$ to mixed states that satisfies strong monotonicity
is known \cite{wg03}, which automatically implies strong monotonicity
of $\G$ on the subset of pure states. In contrast to this, an explicit
counterexample ruling out strong monotonicity is known for $\Gl$, and
since this counterexample considers pure states only \cite{weg04}, no
extension of $\Gl$ to mixed states can be strongly
monotonous. However, we will later see that extensions of $\Gl$ with
weak monotonicity do exist (cf.~$\Gfl$ defined in Section
\ref{subsec:fid}), and therefore $\Gl$ is weakly monotonous on the
subset of pure states.  The axiomatic properties of $\G$ and $\Gl$ are
summarized on the left of Table \ref{GMproperties}.

Regarding the optional axioms, it is easy to verify that $\Gl$ is
normalized whereas $\G$ is not. This makes $\Gl$ the natural choice
for quantitative studies of entanglement, such as scaling laws or
comparison with other measures.  The MES entanglement of $n$ qubits
($n \geq 3$) scales linearly as $n - 2 \log_2 (n) - O (1) < \Gl
(\ket{\Psi}) < n-1$ \cite{jhk08,gfe09,zch10}.  Restricting the
computational coefficients to real values does not affect this scaling
\cite{zch10}, but for positive states (i.e.~states whose coefficients
are all positive in the computational basis) the $n$-qubit MES are
bounded by $\Gl (\ket{\Psi}) \leq \frac{n}{2}$, and this bound is
strict for even $n$ (a trivial example being $n/2$ Bell pairs)
\cite{aulbach12}.  On the other hand, symmetric $n$-qubit MESs scale
logarithmically as $\log_2 (n + 1) - O (1) < \Gl (\ket{\Psi}) < \log_2
(n+1)$ \cite{aulbach12,mgb10}. These scaling laws readily generalize
to qudits, leading to the conclusion that the MESs of sufficiently
large multipartite systems are neither positive nor symmetric.
Furthermore, since generic states are nearly maximally entangled with
respect to GM \cite{gfe09}, the above scaling laws can also be applied
to random states.

Regarding the additivity axioms, neither $\G$ nor $\Gl$ are additive
in general.  For $\G$ this is obvious from its codomain $[0,1]$, and
for $\Gl$ it has been shown that states with a high amount of
entanglement are not additive \cite{zch10}.  Nevertheless, many states
of interest are additive or even strongly additive under $\Gl$.  In
particular, it has been shown that positive states are strongly
additive \cite{zch10}:
\begin{lemma}\label{le:pureadditivity}
  Let $\ket{\psi} \in \cH$ be a positive state.  Then $\ket{\psi}$ is
  strongly additive, i.e. $\Lam (\ket{\psi} \otimes \ket{\phi}) = \Lam
  (\ket{\psi}) \Lam (\ket{\phi})$ and $\Gl (\ket{\psi} \otimes
  \ket{\phi}) = \Gl (\ket{\psi}) + \Gl (\ket{\phi})$ holds for all
  $\ket{\phi} \in \cH$.
\end{lemma}
Examples of positive states are multipartite Dicke states of arbitrary
dimension, and all bipartite pure states (by means of the Schmidt
decomposition). Lemma \ref{le:pureadditivity} can be readily
generalized to mixed states for an extension of GM that will be
discussed in Subsection \ref{subsec:trace}.

Note that the definition of the linear GM \eqref{GMpure1} coincides
with the Groverian entanglement measure $\text{E}_{\text{Gr}} (
\ket{\psi} ) = \G ( \ket{\psi} )^{1/2}$, which has a tangible
operational interpretation by means of a quantum algorithm
\cite{ssb06,bno02}.

\subsection{GM for mixed states}\label{sec:mixed}

With the GM defined for pure states $\ket{\psi} \in \cH$, we now
consider the possible extensions to mixed states $\rho \in \SH$.
Extensions of the linear GM will be labelled $\Gx ( \rho )$, and
extensions of the logarithmic GM as $\Gxl ( \rho )$, where $x$ stands
for a label to denote the extension.  Any valid extension must
coincide with \eqref{GMpure1} or \eqref{GMpure2} on the subset of pure
states $\rho = \proj{\psi}$. In other words, $\Gx ( \proj{\psi} ) = \G
( \ket{\psi} )$ and $\Gxl ( \proj{\psi} ) = \Gl ( \ket{\psi} )$ must
hold for all $\ket{\psi} \in \cH$.

Since the expressions ``pure'' and ``mixed'' can be ambiguous, we
briefly clarify their usage.  From a mathematical viewpoint, $\sigma =
\proj{\psi} \in \SH$ is a mixed state, but physically it is equivalent
to the pure state $\ket{\psi} \in \cH$.  Therefore, we refer to both
$\ket{\psi}$ and $\sigma = \proj{\psi}$ as \emph{pure} states.  On the
other hand, we refer to all states $\rho \in \SH$ as \emph{mixed}
states, so $\sigma = \proj{\psi}$ can be regarded as pure and
mixed. Mixed states that are not pure will be called \emph{genuinely
  mixed}. Mathematically, a state $\rho \in \SH$ is genuinely mixed,
iff $\rank \rho \geq 2$.

One strategy to extend \eqref{GMpure} to mixed states is to extend
\eqref{purelambda} to mixed states, i.e.~to introduce a function
$\Lamx ( \rho ) : \SH \rightarrow \bRp$ with the property that $\Lamx
( \proj{\psi} ) = \Lam ( \ket{\psi} )$ for all $\ket{\psi} \in \cH$.
The following lemma asserts the properties of extensions defined in
that manner.
\begin{lemma}\label{le:weakmonotonicity}
  Let $\Lamx ( \rho ) : \SH \rightarrow \bRp$ be an extension of
  \eqref{purelambda} to the set of all density matrices.  Then the
  following holds for $\Gx ( \rho ) := 1 - \Lamx ( \rho )$ and $\Gxl (
  \rho ) := - \log_2 \Lamx ( \rho )$:
  \begin{enumerate}
  \item $\Gx ( \rho ) \leq \Gxl ( \rho )$ holds for all $\rho$.

  \item $\Gxl ( \rho ) = - \log_2 (1- \Gx ( \rho ))$, or equivalently
    $\Gx ( \rho ) = 1 - 2^{- \Gxl ( \rho )}$ holds for all
    $\rho$. Furthermore, $\Gx \cong \Gxl$.

  \item $\Gx$ satisfies Axiom 1 if and only if $\Gxl$ does.

  \item $\Gx$ satisfies Axiom 2(a) if and only if $\Gxl$ does.

  \item If $\Gx$ is convex, then $\Gxl$ is also convex.

  \item If $\Gxl$ is concave, then $\Gx$ is also concave.
  \end{enumerate}
\end{lemma}
\begin{proof}
  \begin{enumerate}
  \item This directly follows from the elementary inequality
    \eqref{elin}.

  \item The relationships between $\Gx ( \rho )$ and $\Gxl ( \rho )$
    are obtained by eliminating $\Lamx ( \rho )$.  Since $f (y) = -
    \log_2 (1-y)$ and $g (y) = 1 - 2^{-y}$ are monotonously increasing
    functions in $y \in [0,1]$, $\Gx$ and $\Gxl$ have the same
    ordering.

  \item For any $\rho \in \sep$ we have: $\Gx ( \rho ) = 0
    \Leftrightarrow 1 - \Lam ( \rho ) = 0 \Leftrightarrow \Lam ( \rho
    ) = 1 \Leftrightarrow - \log_2 \Lam ( \rho ) = 0 \Leftrightarrow
    \Gxl ( \rho ) = 0 \, .$

  \item For any $\rho \mapsto \sigma = \sum_i \vec{P}_i \rho
    \vec{P}_i^{\dagger}$ we have: $\Gx ( \rho ) \geq \Gx ( \sigma )
    \Leftrightarrow 1 - \Lam ( \rho ) \geq 1 - \Lam ( \sigma )
    \Leftrightarrow \Lam ( \rho ) \leq \Lam ( \sigma ) \Leftrightarrow
    - \log_2 \Lam ( \rho ) \geq - \log_2 \Lam ( \sigma )
    \Leftrightarrow \Gxl ( \rho ) \geq \Gxl ( \sigma ) \, .$

  \item Let $\Gx ( \rho ) = 1 - \Lam ( \rho )$ be convex. Since $g ( y
    ) := - \log_2 (1-y)$ is a convex non-decreasing function, $g ( \Gx (
    \rho )) = - \log_2 \Lam ( \rho )$ is also convex.

  \item Let $\Gxl ( \rho ) = - \log_2 \Lam ( \rho )$ be concave. Its
    additive inverse $- \Gxl ( \rho )$ is therefore convex. Since $g (
    y ) := 2^{y}$ is a convex non-decreasing function, $g ( - \Gxl (
    \rho )) = \Lam ( \rho )$ is also convex. Therefore, $1 - \Lam (
    \rho )$ is concave.
    \qed
  \end{enumerate}
\end{proof}
With regard to items 5.~and 6.~of Lemma \ref{le:weakmonotonicity}, it
should be noted that convexity of $\Gx$ does not follow from convexity
of $\Gxl$, and that concavity of $\Gxl$ does not follow from concavity
of $\Gx$. A counterexample for the latter case are the measures $\Gm$
and $\Gml$ introduced in Section \ref{subsec:trace}.

Several extensions of GM have been proposed in the past, and below we
will introduce these as well as new ones.  The first approach is based
on a convex roof construction, akin to the entanglement of formation
\cite{pv07}.  The second and third approach are based on extending the
definition \eqref{purelambda} to mixed states by means of the fidelity
between the given state and the set of all pure or mixed states,
respectively. Consequently, Lemma \ref{le:weakmonotonicity} applies to
these two approaches.  The fourth approach is to extend the linear GM
by means of the trace distance \eqref{trdist}.

\subsubsection{Extension by convex roof: $\Gc /
  \Gcl$}\label{subsec:cr}

Based on definitions \eqref{purelambda} and \eqref{GMpure}, the convex
roofs of the linear and logarithmic GM are
\begin{subequations}\label{ea:convexroof}
  \begin{align}
    \label{ea:convexroof1}
    \Gc ( \rho ) := \min_{ \{ p_{i}, \ket{\psi_{i}} \} }
    \sum_{i} p_{i} \G \left( \ket{\psi_{i}} \right) \, , \\
    \label{ea:convexroof2}
    \Gcl ( \rho ) := \min_{ \{ p_{i}, \ket{\psi_{i}} \} }
    \sum_{i} p_{i} \Gl \left( \ket{\psi_{i}} \right) \, ,
  \end{align}
\end{subequations}
where the minimum runs over all decompositions of $\rho = \sum_{i}
p_{i} \proj{\psi_{i}}$.  Decompositions that maximize
\eqref{ea:convexroof1} or \eqref{ea:convexroof2} will be called
\emph{optimal decompositions}, and labeled $\{ P_i , \ket{\Psi_i} \}$.
It is natural to ask whether for a given $\rho$ the two functionals
are minimized for the same decomposition.  We will later show that for
many states $\Gc$ and $\Gcl$ have the same optimal decompositions
(e.g. for all isotropic states and two qubit states), but that there
also exist states for which $\Gc$ and $\Gcl$ do not have a common
optimal decomposition (e.g. for some maximally correlated
states). Another open question is how many pure components
$\ket{\psi_i}$ are necessary for an optimal decomposition of $\Gc$ or
$\Gcl$. At least for $\Gc$ it is known that $(\dim \cH )^2$ pure
components suffice \cite{skb11NJP}.

Mathematically, any two decompositions $\{ p_i , \ket{\psi_i} \}$ and
$\{ q_j , \ket{\phi_j} \}$ of the same $\rho$ are related by a unitary
matrix $u_{ij}$ (i.e.~$\sum_k u_{ki}^{*} u_{kj}=\delta_{ij}$), so that
\begin{equation}\label{unitaryconversion}
  \sqrt{p_i} \ket{\psi_{i}} = \sum_j u_{ij} \sqrt{q_j} \ket{\phi_j}
\end{equation}
holds for all $i$ \cite{hjw93,nc00book}.  This identity will later be
used in some proofs.

Regarding the entanglement axioms, Axiom 1 and convexity directly
follow from the convex roof definitions for both $\Gc$ and $\Gcl$.
The quantity $\Gc$ was first studied in detail in the seminal paper of
Wei \etal \cite{wg03}, and found to be a strong entanglement measure.
On the other hand, $\Gcl$ cannot be a strong entanglement measure
\cite{weg04}.  However, it is an open question whether $\Gcl$ is
weakly monotonous, i.e.~whether $\Gcl(\rho) \geq \Gcl(\channel
(\rho))$ holds for all channels $\channel$ and all $\rho \in \SH$. For
$\rho \in \sep$ this is satisfied because of Axiom 1, and for all
bipartite pure MES $\ket{\Psi} = \frac{1}{\sqrt{d}} \sum_{i=1}^{d}
\ket{ii}$ this is also satisfied, because
\begin{equation*}
  \Gcl (\channel (\ket{\Psi})) =
  \min_{ \{ p_i , \ket{\psi_i} \} } \sum_i p_i \Gl (\ket{\psi_i}) \leq
  \min_{ \{ p_i , \ket{\psi_i} \} } - \sum_i p_i \log \tfrac{1}{d} =
  \log d = \Gcl (\ket{\Psi}) \, .
\end{equation*}
In the two qubit case $\Gcl$ is a weak entanglement measure.  In
Appendix \ref{ap:weak} we present a direct proof employing the
concurrence and the optimal decomposition of the entanglement of
formation \cite{wootters98}.  The same result will later follow from a
different line of argumentation.

Since the definitions \eqref{ea:convexroof} are not based on extending
\eqref{purelambda} to mixed states, Lemma \ref{le:weakmonotonicity}
does not apply to $\Gc$ and $\Gcl$. In particular, we do not know if
there exists an exact analytic relation between $\Gc$ and $\Gcl$, or
if the two quantities have at least the same ordering (i.e.~$\Gc \cong
\Gcl$).  The following lemma provides an analytic relation in the form
of an inequality.
\begin{lemma}\label{gcandgcl}
  For every state $\rho$ the following holds.
  \begin{enumerate}
  \item $\Gc ( \rho ) \leq \Gcl ( \rho )$.

  \item $\Gcl ( \rho ) \geq - \log_2 (1- \Gc ( \rho ))$, or
    equivalently $\Gc ( \rho ) \leq 1 - 2^{- \Gcl ( \rho )}$.
  \end{enumerate}
\end{lemma}
\begin{proof}
  \begin{enumerate}
  \item This inequality readily follows from \eqref{ea:Gg<=Gl} and
    \eqref{ea:convexroof} for all $\rho$.

  \item Let $\{ P_i , \ket{\Psi_i} \}$ be an optimal decomposition of
    $\rho$ for $\Gcl (\rho )$. Then
    \begin{equation*}
      \begin{split}
        \Gcl (\rho) &=
        \sum_i P_i \big[ - \log \Lam (\ket{\Psi_i}) \big] \geq
        - \log \Big[ \sum_i P_i \Lam (\ket{\Psi_i}) \Big] \\
        &\geq
        - \log \Big[ \max_{ \{ p_i , \ket{\psi_i} \} }
        \sum_i p_i \Lam (\ket{\psi_i}) \Big]
        = - \log \big[ 1-  \Gc ( \rho ) \big] \, ,
      \end{split}
    \end{equation*}
    where the inequalities follow from that fact that $f(y) = - \log_2
    (y)$ is a convex and monotonically decreasing function.  An
    analogous derivation yields $\Gc ( \rho ) \leq 1 - 2^{- \Gcl (
      \rho )}$.
     \qed
  \end{enumerate}
\end{proof}
Regarding the ordering of $\Gc$ and $\Gcl$, it will later be shown
that the two measures do not have the same ordering in general
(cf.~Corollary \ref{ineqorder}).

\subsubsection{Extension by trace inner product: $\Gm /
  \Gml$}\label{subsec:trace}

Apart from the convex roof, the most widely studied extension of GM to
mixed states is obtained by extending \eqref{purelambda} to mixed
states via the Hilbert-Schmidt inner product $\tr (A^{\dagger} B)$,
also known as trace inner product \cite{nc00book}.  As proved below,
this is equivalent to maximizing the fidelity between the input state
and set of pure product states.
\begin{equation}\label{eq:EE}
  \Lamm (\rho) := \max_{\sigma \in \sep} \tr ( \rho \sigma )
  = \max_{\ket{\varphi} \in \pro}  \bra{\varphi} \rho \ket{\varphi}
  = \max_{\ket{\varphi} \in \pro} F^2 ( \rho , \ket{\varphi} ) \, .
\end{equation}
\begin{proof}
  The last equality is clear from \eqref{pf1}, and the $\ge$ part of
  the middle equality follows from the fact that the set of fully
  separable states contains the pure product states.  The $\le$ part
  of the middle equation follows as
  \begin{multline*}
    \max_{\sigma \in \sep} \tr ( \rho \sigma ) =
    \tr \Big( \rho \sum_i p_i \proj{\phi_i} \Big)
    = \sum_i p_i \bra{\phi_i} \rho \ket{\phi_i} \\
    \leq \sum_i p_i \max_{\ket{\varphi} \in \pro}
    \bra{\varphi} \rho \ket{\varphi} =
    \max_{\ket{\varphi} \in \pro}
    \bra{\varphi} \rho \ket{\varphi} \, ,
  \end{multline*}
  where $\sigma_{\text{m}} = \sum_i p_i \proj{\phi_i}$, with
  $\ket{\phi_i} \in \pro$ for all $i$, is the separable state that
  maximizes $\tr ( \rho \sigma )$.
  \qed
\end{proof}

For pure states $\Lamm (\rho)$ obviously coincides with $\Lam (
\ket{\psi} )$ from \eqref{purelambda}. Therefore, the functionals
\eqref{GMpure1} and \eqref{GMpure2} are extended to mixed states as
\begin{subequations}\label{GMmixed}
  \begin{align}
    \Gm ( \rho ) &:= 1 - \Lamm ( \rho ) \, ,
    \label{GMmixed1} \\
    \Gml ( \rho ) &:= - \log \Lamm ( \rho ) \, .
    \label{GMmixed2}
  \end{align}
\end{subequations}
Since Lemma \ref{le:weakmonotonicity} applies to these measures, there
is some interdependence in their entanglement axioms.  Indeed, neither
$\Gm$ nor $\Gml$ is an entanglement measure.  This can be readily seen
from that fact that the two measures attain their maximum for the
maximally mixed state $\frac{\one}{\dim ( \cH )}$, a separable state
\cite{hmm06}, thus violating Axiom 1 and Axiom 2(a).

With regard to convexity, the maximally mixed state is also a
counterexample, because it can be decomposed into pure product states,
$\one = \sum_i \proj{i}$, with $\G ( \ket{i} ) = \Gl ( \ket{i} ) =
0$. To check whether $\Gm$ or $\Gml$ are concave, consider an
arbitrary decomposition $\rho = \sum_i p_i \rho_i$. We have
\begin{equation*}
  \Lamm ( \rho ) = \max_{\ket{\varphi} \in \pro}
  \sum_i p_i \bra{\varphi} \rho_i \ket{\varphi} \leq
  \sum_i p_i \max_{\ket{\varphi_i} \in \pro} \bra{\varphi_i}
  \rho_i \ket{\varphi_i}
  = \sum_i p_i \Lamm ( \rho_i ) \, .
\end{equation*}
Therefore, $\Lamm ( \rho )$ is convex.  From this it directly follows
that $\Gm ( \rho ) = 1 - \Lamm ( \rho )$ is concave. It remains to
investigate whether $\Gml$ is also concave. Let us consider the
isotropic state $\rho_{\iso} := p \one /d^2 + (1-p) \proj{\Psi}$,
where $p \in [0,1]$ and $\ket{\Psi} = \frac{1}{\sqrt d} \sum^d_{i=1}
\ket{ii}$ \cite{hh99}. We easily see that $\Lamm(\rho_{\iso}) = p/d^2
+ (1-p)/d$. The concavity of the logarithm yields $\Gml(\rho_{\iso})
\leq p \Gml( \one /d^2) + (1-p) \Gml(\ket{\Psi})$, and the inequality
is strict for $p \in (0,1)$. So, the isotropic state is a
counterexample to the concavity of $\Gml$. To conclude, $\Gml$ is
neither convex nor concave.

Although not entanglement measures, the quantities $\Gm$ and $\Gml$
are easier to calculate than other definitions of GM, and have
received a considerable amount of attention
\cite{hmm06,mmv07,hmm06,hmm08,wei08,jhk08,zch10}.  The measure $\Gml$
has been found to be closely related to the relative entropy of
entanglement and the logarithmic global robustness of entanglement
\cite{weg04,hmm08,aulbach12,zch10}. Furthermore, $\Gml$ has been
employed for the construction of optimal entanglement witnesses
\cite{hmm08} and for the study of state discrimination under LOCC
\cite{hmm08,hmm06}.  Zhu \etal \cite{zch10} calculated $\Lamm ( \rho
)$ for many states of interest, and Jung \etal \cite{jhk08} found that
tracing out one subsystem of an $n$-partite pure state does not change
this quantity, i.e.~$\Lamm (\ket{\psi}) = \Lamm (\rho)$, with $\rho =
\tr_i (\proj{\psi})$, holds for all $\ket{\psi} \in \cH$ and all $1
\leq i \leq n$.

The measure $\Gml$ allows to generalize Lemma \ref{le:pureadditivity}
to mixed states \cite{zch10}.  A density matrix is called positive if
all its entries in the computational basis are positive.
\begin{lemma}\label{le:mixedadditivity}
  Let $\rho \in \SH$ be a positive state.  Then $\rho$ is strongly
  additive, i.e. $\Lam (\rho \otimes \sigma) = \Lam (\rho) \Lam
  (\sigma)$ and $\Gml (\rho \otimes \sigma) = \Gml ( \rho ) + \Gml (
  \sigma )$ holds for all $\sigma \in \SH$.
\end{lemma}
This lemma has been employed to show the strong additivity of many
mixed states, such as mixtures of Dicke states, Bell diagonal states,
isotropic states, multiqubit D\"{u}r states and the Smolin state
\cite{zch10}.  The additivity problem of $\Gml$ is closely related to
that of the relative entropy and the logarithmic global robustness
\cite{hmm08}, which facilitated the study of additivity under these
two entanglement measures, as well \cite{zch10}.

\subsubsection{Extension by fidelity: $\Gf / \Gfl$}\label{subsec:fid}

We can also extend $\Lam$ to mixed states by means of the fidelity:
\begin{equation}\label{LamFid}
  \Lamf ( \rho ) := \max_{\sigma \in \sep}
  F^2 ( \rho , \sigma ) \, .
\end{equation}
This quantity has been previously studied
\cite{bno02,ssb06,cw07,skb11NJP,skb11PRA,hmw13}, and has been
described as \emph{fidelity of separability} \cite{skb11NJP}.  In the
bipartite case \eqref{LamFid} is equivalent to the so-called maximum
$k$-extendible fidelity of a state in the limit $k \rightarrow \infty$
\cite{hmw13}. The maximum $k$-extendible fidelity has an operational
interpretation as the maximum probability with which one party can
convince another party that $\rho$ is separable in a specific protocol
\cite{hmw13}.

It is easy to verify that for pure states $\Lamf$ coincides with
$\Lam$:
\begin{equation*}
  \Lamf (\ket{\psi}) =
  \max_{\sigma \in \sep} F^2 ( \ket{\psi} , \sigma ) =
  \max_{\sigma \in \sep}
  \abs{\bra{\psi} \sigma \ket{\psi}}
  = \max_{\ket{\varphi} \in \pro}
  \abs{\braket{\varphi}{\psi}}^2
  = \Lam ( \ket{\psi} ) \, .
\end{equation*}
The corresponding extensions of the linear and logarithmic GM are
\begin{subequations}\label{ea:GMfidel}
  \begin{align}
    \Gf ( \rho ) &:= 1 - \Lamf ( \rho ) \, , \label{GMfidel1} \\
    \Gfl ( \rho ) &:= - \log \Lamf ( \rho ) \, . \label{GMfidel2}
  \end{align}
\end{subequations}
As seen for $\Gm$ and $\Gml$, Lemma \ref{le:weakmonotonicity} applies
to these measures.  $\Gf$ is intimately related to the Groverian
entanglement measure \cite{bno02,ssb06}, thus giving it an operational
interpretation by a quantum algorithm.  $\Gf$ has been shown to be a
weak entanglement measure \cite{ssb06}, and has also been studied in
\cite{cw07}.  On the other hand, little is known about $\Gfl$. It has
been touched upon in the context of additivity in \cite{skb11PRA}, but
to our knowledge its properties have not been studied before.

Intriguingly, it was discovered that $\Gf ( \rho )$ is equivalent to
its convex roof \cite{skb11NJP}, and since the convex roof is
precisely $\Gc ( \rho )$, the definitions \eqref{GMfidel1} and
\eqref{ea:convexroof1} are equivalent:
\begin{proposition}\label{prop:Gc=Gf}
  $\Gf \equiv \Gc$, i.e.~$\Gf( \rho ) = \Gc( \rho )$ holds for all
  states $\rho$.
\end{proposition}
We will jointly refer to these two definitions as $\Gfc$ in the
following, except in situations where the emphasis is on their formal
definitions (i.e.~fidelity-based vs.~convex roof-based).  The
relationship between $\sigma_{\text{f}}$, the CSS of $\rho$ in terms
of \eqref{GMfidel1}, and the optimal decomposition $\{ P_i ,
\ket{\Psi_i} \}$ of $\rho$ in terms of \eqref{ea:convexroof1} is also
fully understood, and outlined in \cite{skb11NJP}.  Since $\Gc$ is
known to be a strong entanglement measure with convexity, the same is
true for $\Gf$. In particular, the convexity of \eqref{GMfidel1}
implies that $\Lamf ( \rho )$ is concave.  There are many states for
which $\Gc ( \rho )$ has been computed \cite{wg03}, and from
Proposition \ref{prop:Gc=Gf} and Lemma \ref{le:weakmonotonicity} the
values of $\Gf ( \rho )$ and $\Gfl ( \rho )$ directly follow.

With the known properties of $\Gf$, it follows from Lemma
\ref{le:weakmonotonicity} that $\Gfl$ is a weak entanglement measure
with convexity.  From the convexity of $\Gfl$ it then follows that
$\Gfl ( \rho ) \leq \Gcl ( \rho )$ holds for all $\rho$:
\begin{equation}\label{ea:convexineq}
  \Gfl ( \rho ) \leq
  \sum_i P_{i} \Gl \left( \ket{\Psi_{i}} \right) =
  \min_{ \{ p_{i}, \ket{\psi_{i}} \} } \sum_{i} p_{i}
  \Gl \left( \ket{\psi_{i}} \right)
  = \Gcl ( \rho ) \, ,
\end{equation}
where $\{P_i, \ket{\Psi_{i}}\}$ is an optimal decomposition for $\Gcl
( \rho )$.  The question whether $\Gfl ( \rho ) \leq \Gcl ( \rho )$ is
strict for some states will be extensively studied in Section
\ref{sec:relationships}.

\subsubsection{Extension by trace distance: $\Gt$}\label{sub:gt}

It is tempting to introduce another mixed extension of GM, based on
the trace distance defined in \eqref{trdist}.  From \eqref{tf3} we
obtain $\DT^2 ( \ket{\psi} , \ket{\phi} ) = 1 - F^2 ( \ket{\psi} ,
\ket{\phi} ) = 1 - \abs{\braket{\psi}{\phi}}^2$, an expression with
the form of \eqref{GMpure1}.  We therefore define
\begin{equation}\label{ea:Gtr}
  \Gt ( \rho ) := \min_{\ket{\varphi} \in \pro}
  \DT^2 ( \rho , \ket{\varphi} ) =
  \frac{1}{4} \min_{\ket{\varphi} \in \pro}
  \big( \tr \big\vert \rho - \proj{\varphi} \big\vert \big)^2 \, ,
\end{equation}
and call this measure the \emph{trace extension of GM}.  For pure
input states $\Gt$ obviously coincides with \eqref{GMpure1}, so
\eqref{ea:Gtr} is an extension of the linear GM.  Note that the
related definition,
\begin{equation}\label{ea:GTpET}
  \ETw ( \rho ) =
  \min_{\sigma \in \sep} \DT^2 ( \rho , \sigma ) \, ,
\end{equation}
was already introduced (up to a square operation) as the trace
entanglement in \eqref{tr1} and shown to be a weak entanglement
measure.  From \eqref{ea:Gtr} and \eqref{ea:GTpET} it immediately
follows that $\ETw ( \rho ) \leq \Gt ( \rho )$ for all $\rho$.  To see
whether $\ETw$ also coincides with \eqref{GMpure1} for the subset of
pure states, we need to answer the question whether for pure input
states the closest separable state $\sigma_{\text{t}} \in \SH$ in
terms of the trace distance can always be chosen to be pure,
i.e. $\sigma_{\text{t}} = \ketbra{\phi}{\phi}$.  The cluster states
provide a counterexample for this:
\begin{corollary}\label{cor:traceent}
  There exist pure states $\ket{\psi}$ for which $\ETw ( \ket{\psi} )
  < \Gt ( \ket{\psi} )$ holds.
\end{corollary}
\begin{proof}
  From Theorem \ref{th:traceent} and the succeeding paragraph,
  together with \eqref{tf3} it follows that $\ETw ( \ket{\text{C}_n} )
  = \DT^2 ( \ket{\text{C}_n} , \delta ) = (1- 2^{-\frac{n}{2}})^2 < 1
  - 2^{-\frac{n}{2}} = 1 - \Lam ( \ket{\text{C}_n} ) = \Gt (
  \ket{\text{C}_n} )$ holds for all $n$ qubit cluster states
  $\ket{\text{C}_n}$ with even $n$.
  \qed
\end{proof}
As a consequence, $\ETw$ is not an extension of the linear GM, because
$\DT ( \ket{\psi} , \cdot )$ is in general not minimized by pure
states.  However, $\ETw$ is an interesting quantity on its own,
because it is a weak entanglement measure and a lower bound to $\Gt$.
Furthermore, we will see in Section \ref{relgtet} that $\ETw$ is also
a lower bound to $\Gfc$, which makes it a joint lower bound to all the
GM definitions discussed in this paper.

The convexity of $\ET$, and thus the convexity of $\ETw = \ET^2$ can
be proved with the joint convexity of the trace distance. For any
$\rho = \sum_i p_i \rho_i$ we have
\begin{multline*}
  \sum_i p_i \ET ( \rho_i ) =
  \sum_i p_i \! \min_{\sigma_i \in \sep} \! \DT ( \rho_i , \sigma_i )
  = \! \min_{\{ \sigma_i \} \in \sep}
  \bigg[ \sum_i p_i \DT ( \rho_i , \sigma_i ) \bigg] \\
  \geq
  \min_{ \{ \sigma_i \} \in \sep}
  \Big[ \DT \big( \rho , {\textstyle \sum_i} p_i \sigma_i \big) \Big]
  = \min_{\sigma \in \sep} \DT ( \rho , \sigma )
  = \ET ( \rho ) \, ,
\end{multline*}
where the inequality follows from \eqref{trjointconvex}.

On the other hand, $\Gt$ is not an entanglement measure.  For this,
note that $\Gt ( \ket{\varphi} ) = 0$ holds for all $\ket{\varphi} \in
\pro$, and that $\Gt ( \rho ) > 0$ holds for all genuinely mixed $\rho
\in \sep$.  From this it is not only clear that Axiom 1 is violated,
but one can also immediately construct counterexamples for the
concavity and the weak monotonicity (e.g.~with the depolarizing
channel \cite{nc00book}).  Using the isotropic state as a
counterexample, it is shown in Appendix \ref{ap:gtpconcave} that $\Gt$
is not concave either.

\section{Relationships between the GM
  definitions}\label{sec:relationships}

In the previous section we introduced and discussed seven different
definitions of GM, of which two are equivalent. In the following we
analyze the relationship between these different definitions.

For an arbitrary mixed state $\rho \in \SH$ the quantities $\Gm ( \rho
)$ and $\Gml ( \rho )$ correspond to the same closest product state
$\phim \in \pro$, and the quantities $\Gf ( \rho )$ and $\Gfl ( \rho
)$ correspond to the same closest separable state $\sigma_{\text{f}}
\in \sep$.  In contrast to this, $\Gc ( \rho )$ and $\Gcl ( \rho )$
correspond to optimal decompositions $\{ P_i , \ket{\Psi_i} \}$ that
may be different for the two measures.  The quantity $\Gt ( \rho )$
corresponds to a closest product state $\ket{\phi_{\text{t}}} \in
\pro$.  In total, with the exception of the convex roof-based
measures, the different values of GM for a given state are determined
by two product states $\phim$, $\ket{\phi_{\text{t}}}$, and one
separable state $\sigma_{\text{f}}$.

\subsection{Comparison between $\Gfc$, $\Gcl$ and $\Gfl$}

From the previous discussion we already know that $\Gf \equiv \Gc$,
$\Gfc \cong \Gfl$, and that $\Gcl (\rho ) \geq \Gfl (\rho) = - \log_2
\big( 1- \Gfc (\rho ) \big)$.  In this subsection, we study the
connection between the fidelity-based and convex roof-based extensions
in further detail by addressing some open problems. For example, it is
neither obvious nor known whether $\Gcl$ and $\Gfl$ are equivalent
($\Gcl \equiv \Gfl$), and if not, whether they have at least the same
ordering ($\Gcl \cong \Gfl$).  For this purpose, we will first derive
necessary and sufficient conditions for $\Gcl ( \rho ) = \Gfl ( \rho
)$, and then investigate optimal decompositions for specific classes
of states (e.g.~maximally correlated states, isotropic states, 2 qubit
states).
\begin{theorem}\label{thm:Gfl=Gcl}
  For any state $\rho$ the following four conditions are equivalent:
  \begin{enumerate}
  \item $\Gcl (\rho ) = \Gfl (\rho )$ holds.

  \item $\Gcl (\rho ) = - \log_2 \big( 1- \Gfc (\rho ) \big)$ holds.

  \item There exists a decomposition $\{P_i, \ket{\Psi_{i}}\}$, so
    that
    \begin{enumerate}
    \item $\{P_i, \ket{\Psi_{i}}\}$ is optimal for $\Gc ( \rho )$ and
      $\Gcl ( \rho )$, and

    \item the $\ket{\Psi_i}$ are all equally entangled.
    \end{enumerate}

  \item For every optimal decomposition $\{P_i, \ket{\Psi_{i}}\}$ of
    $\Gcl ( \rho )$ the following holds
    \begin{enumerate}
    \item $\{P_i, \ket{\Psi_{i}}\}$ is also optimal for $\Gc ( \rho
      )$, and

    \item the $\ket{\Psi_i}$ are all equally entangled.
    \end{enumerate}
  \end{enumerate}
  Here, the meaning of 3(b) and 4(b) is that $\Lam (\ket{\Psi_i}) =
  \Lam(\ket{\Psi_j})$ holds for all $i,j$.
\end{theorem}
\begin{proof}
  Let $\{ P_i , \ket{\Psi_i} \}$ be some optimal decomposition of
  $\Gcl (\rho)$.  Using Lemma \ref{le:weakmonotonicity}, Proposition
  \ref{prop:Gc=Gf} and \eqref{ea:convexroof}, we have
  \begin{equation}\label{ea:Gfl=Gcl}
    \begin{split}
      \Gfl ( \rho ) &=
      - \log \big[ 1- \Gfc ( \rho ) \big] =
      -\log \Big[ \max_{\{p_i, \ket{\psi_i} \}}
      \sum_i p_i \Lam ( \ket{\psi_{i}} ) \Big] \\
      &\leq
      -\log \Big[ \sum_i P_i \Lam (\ket{\Psi_{i}}) \Big]
      \leq -\sum_i P_i \log \Lam (\ket{\Psi_{i}})
      = \Gcl ( \rho ) \, ,
    \end{split}
  \end{equation}
  where the second inequality follows from the concavity of the
  logarithm.
  \begin{itemize}
  \item $1. \Leftrightarrow 2.$: This equivalency follows immediately
    from $\Gfl (\rho ) = - \log_2 \big( 1- \Gfc (\rho ) \big)$.

  \item $1. \Leftrightarrow 4.$: Apparently, $\Gfl ( \rho ) = \Gcl (
    \rho )$ holds if and only if both inequalities in
    \eqref{ea:Gfl=Gcl} become equalities.  Regarding the first
    inequality in \eqref{ea:Gfl=Gcl}, this inequality becomes an
    equality iff $\{ P_i , \ket{\Psi_i} \}$ is also an optimal
    decomposition for $\Gc ( \rho )$.  Regarding the second inequality
    in \eqref{ea:Gfl=Gcl}, the strict concavity of the logarithm
    implies that this inequality becomes an equality iff $\Lam
    (\ket{\Psi_{i}})=\Lam (\ket{\Psi_{j}})$ holds for all
    $i,j$. Together, this yields that condition 1. holds if and only
    if condition 4. holds.

  \item $4. \Rightarrow 3.$: Obvious, since $\Gcl (\rho)$ has at least
    one optimal decomposition.

  \item $3. \Rightarrow 1.$: Using the decomposition $\{ P_i ,
    \ket{\Psi_i} \}$ postulated by condition 3, the two inequalities
    in \eqref{ea:Gfl=Gcl} turn into equalities. Therefore, $\Gfl (
    \rho ) = \Gcl ( \rho )$.
    \qed
  \end{itemize}
\end{proof}
Note that for states $\rho$ that fall under Theorem \ref{thm:Gfl=Gcl}
any optimal decomposition for $\Gcl( \rho )$ is also optimal for $\Gc(
\rho )$, but the converse is not necessarily true.  In other words,
for states $\rho$ that satisfy $\Gcl( \rho ) = \Gfl ( \rho )$, the set
of optimal decompositions for $\Gcl ( \rho )$ is a non-empty subset of
the set of optimal decompositions for $\Gc ( \rho )$. It is an open
question whether states $\rho$ with $\Gcl( \rho ) = \Gfl ( \rho )$
exist, for which the set of optimal decompositions of $\Gc ( \rho )$
is strictly larger than that of $\Gcl ( \rho )$.  Without the
condition $\Gcl( \rho ) = \Gfl ( \rho )$, $\Gc ( \rho )$ and $\Gcl (
\rho )$ may not even have any common optimal decomposition, as shown
later in Corollary \ref{exclusiv}.

The following corollary helps to understand the relationship between
optimal decompositions of $\Gc ( \rho )$ and $\Gcl ( \rho )$.
\begin{corollary}\label{coroptgcgcl}
  Let $\{ P_i ,\ket{\Psi_i} \}$ be some decomposition of $\rho$.  The
  following two conditions are equivalent:
  \begin{enumerate}
  \item $\{ P_i ,\ket{\Psi_i} \}$ is an optimal decomposition for $\Gc
    ( \rho )$, and the $\ket{\Psi_i}$ are all equally entangled.

  \item $\{ P_i ,\ket{\Psi_i} \}$ is an optimal decomposition for
    $\Gcl ( \rho )$, and $\Gcl ( \rho ) = \Gfl ( \rho )$ holds.
  \end{enumerate}
\end{corollary}
\begin{proof}
  \begin{itemize}
  \item 2 $\Rightarrow$ 1: This easily follows from items 1. and 4. of
    Theorem \ref{thm:Gfl=Gcl}.

  \item 1 $\Rightarrow$ 2: Let $\{ P_i , \ket{\Psi_i} \}$ be an
    optimal decomposition of $\Gc ( \rho )$ where the $\ket{\Psi_i}$
    are all equally entangled.  In analogy to \eqref{ea:Gfl=Gcl}, we
    have
    \begin{equation}\label{ea:Gfl=Gcl2}
      \begin{split}
        \Gfl ( \rho ) &= - \log \big[ 1- \Gfc ( \rho ) \big]
        =
        -\log \Big[ \sum_i P_i \Lam (\ket{\Psi_{i}}) \Big] \\
        &= \sum_i P_i \big[ - \log \Lam (\ket{\Psi_{i}}) \big]
        \geq \min_{\{p_i, \ket{\psi_i} \}}
        \sum_i p_i \Gl ( \ket{\psi_{i}} ) =
        \Gcl ( \rho ) \, ,
      \end{split}
    \end{equation}
    where the third equality follows from the fact that the
    $\ket{\Psi_i}$ are all equally entangled.  As shown in
    \eqref{ea:convexineq}, $\Gfl ( \rho ) \leq \Gcl ( \rho )$ holds
    for all $\rho$, so the inequality in \eqref{ea:Gfl=Gcl2} must be
    an equality. Therefore, $\{ P_i ,\ket{\Psi_i} \}$ is also an
    optimal decomposition for $\Gcl ( \rho )$, and $\Gfl ( \rho ) =
    \Gcl ( \rho )$ holds.
    \qed
  \end{itemize}
\end{proof}
Following the derivation of general results, we next investigate
classes of states for whom $\Gfl ( \rho )$ and $\Gcl ( \rho )$
coincide. We will see that for two-qubit systems $\Gfl \equiv \Gcl$
holds, and that for general bipartite systems $\Gfl ( \rho ) = \Gcl (
\rho )$ holds for the subset of isotropic states.
\begin{proposition}\label{pr:fctwoqubits}
  $\Gfl \equiv \Gcl$ holds for $\bC^2 \otimes \bC^2$.
\end{proposition}
\begin{proof}
  According to Proposition 4 of \cite{skb11NJP}, if $f (x)$ is a
  non-negative convex function for $x \geq 0$ and obeys $f (0) = 0$,
  then for two-qubit systems $f ( \Gc ( \rho ))$ is equal to its
  convex roof.  The function $f(x) := - \log_2 (1-x)$ satisfies the
  requirements, so $f ( \Gc ( \rho )) = f ( \Gf ( \rho )) = \Gfl (
  \rho )$ is equal to its convex roof, which is precisely $\min_{ \{
    p_{i}, \ket{\psi_{i}} \} } \sum_{i} p_{i} \Gl \left(
    \ket{\psi_{i}} \right) = \Gcl ( \rho )$. Therefore, $\Gfl ( \rho )
  = \Gcl ( \rho )$ holds for all $\rho$.
  \qed
\end{proof}

\begin{proposition}\label{pr:isotropic}
  $\Gfl ( \rho ) = \Gcl ( \rho )$ holds for isotropic states in $\bCd
  \otimes \bCd$.
\end{proposition}
\begin{proof}
  The isotropic states are $\rho_{\iso} = p \one /d^2 + (1-p)
  \proj{\Psi}$, with $p \in [0,1]$ and $\ket{\Psi} = \frac{1}{\sqrt d}
  \sum^d_{i=1} \ket{ii}$.  The parametrization employed in \cite{wg03}
  translates to ours as $p=\frac{d^2}{d^2-1}(1-F)$, where $F\in
  [0,1]$.  For $F \in [0,\frac{1}{d}]$, $\rho$ is separable
  \cite{wg03}, which implies $\Gfl ( \rho ) = \Gcl ( \rho ) =0$.
  Next, consider the entangled region $F\in (\frac{1}{d},1]$.  From
  \cite{wg03} (Eq.~(36), (54)) it follows that $E_{\sin^2} = \Gfc$ and
  $\Gfc ( \rho_{\iso}(F) )=1- \frac{1}{d}(\sqrt F + \sqrt{(d-1)(1-F)}
  )^2$.  From Lemma \ref{le:weakmonotonicity} it follows that $\Gfl (
  \rho_{\iso}(F) )=-\log \frac{1}{d}(\sqrt F + \sqrt{(d-1)(1-F)}
  )^2$. To compute $\Gcl ( \rho_{\iso}(F) )$, we follow the idea of
  \cite{wg03} to obtain the first and second equalities of the
  following equation:
  \begin{equation}\label{ea:Gcl=Gfl}
    \begin{split}
      \Gcl ( \rho_{\iso}(F) )
      &= \cC_{\text{conv}}
       \bigg[ - \max_{\{\mu_i\}} \Big\{ \log \mu_i \Big\vert F =
        \frac{1}{d} \Big( \sum^d_{i=1}
        \sqrt{\mu_i} \Big)^2 \Big\} \bigg] \\
      &= \cC_{\text{conv}}
      \left [ - \log \frac{1}{d}
        (\sqrt{F} + \sqrt{(d-1)(1-F)} )^2 \right] \\
      &= \cC_{\text{conv}} \left [ \Gfl ( \rho_{\iso}(F) ) \right]
      = \Gfl ( \rho_{\iso}(F) ) \, .
    \end{split}
  \end{equation}
  Here $\cC_{\text{conv}}$ is the convex hull. The last equality in
  \eqref{ea:Gcl=Gfl} follows from the fact that $ \Gfl ( \rho_{\iso}(F)
  )$ is convex. This completes the proof.
  \qed
\end{proof}
Next, we investigate maximally correlated states, because necessary
and sufficient conditions for $\Gcl ( \rho ) = \Gfl ( \rho )$ can be
derived for these states.  The maximally correlated states have been
extensively studied in terms of the entanglement of formation and
distillable entanglement \cite{rains01,dw05}. In this paper we focus
on a special type of maximally correlated states defined as follows:
\begin{definition}\label{defmaxcorr}
  Given a bipartite $d$-level system $\cH = \bCd \otimes \bCd$, let
  $\{ n_i \}_{i=0}^{r}$ with $r < d$ be a set of integers with $0 =
  n_0 < \cdots < n_r = d$, and let $\ket{\Theta_i} =
  \frac{1}{\sqrt{n_i-n_{i-1}}} \sum^{n_i}_{k=n_{i-1}+1} \ket{kk}$ be
  the $({n_i - n_{i-1}})$-level MES.  Then, $\rho = \sum^{r}_{i=1} q_i
  \proj{\Theta_i}$, with $\sum^r_{i=1} q_i =1$ and $q_i \in (0,1)$, is
  a $d \times d$ maximally correlated state.
\end{definition}
From now on, we refer to maximally correlated states as the states
defined above. The integer $r$ can be readily identified as the rank
of the maximally correlated state, i.e. $\rank \rho = r$.

\begin{lemma}\label{le:Gcl>Gfl}
  Let $\rho = \sum^{r}_{i=1} q_i \proj{\Theta_i}$ be a $d \times d$
  maximally correlated state. Then,
  \begin{enumerate}
  \item The unique optimal decomposition for $\Gc( \rho )$, up to
    overall phases, is $\{q_i, \ket{\Theta_i} \}$.
    
  \item $\Gc( \rho ) = \Gf ( \rho ) = 1 - \sum^r_{i=1}
    \frac{q_i}{n_i-n_{i-1}}$ .
    
  \item $\Gcl( \rho ) \geq \Gfl( \rho ) = - \log \big( \sum^r_{i=1}
    \frac{q_i}{n_i-n_{i-1}} \big)$ .
    
  \item $\Gcl( \rho ) = \Gfl( \rho )$ iff the $\ket{\Theta_i}$ are all
    equally entangled.

  \item $\Gcl( \rho ) = \Gfl( \rho )$ iff ${n_i-n_{i-1}}$ is the same
    for all $i=1, \ldots ,r$ .
    
  \item If $\Gcl( \rho ) = \Gfl( \rho )$, then $\Gcl (\rho)$ has the
    same optimal decompositions as $\Gc (\rho)$.
  \end{enumerate}
\end{lemma}
\begin{proof}
  First we prove item 1. Let $\{P_i, \ket{\Psi_{i}} \}$ be some
  optimal decomposition of $\Gc ( \rho )$.  Then
  \begin{equation}\label{ea:Gcl>Gfl}
    1 - \Gc( \rho ) =
    \sum_i P_i \Lam (\ket{\Psi_{i}}) \geq
    \sum^r_{j=1} q_i \Lam ( \ket{\Theta_{j}} ) =
    \sum^r_{j=1} \frac{q_j}{n_j - n_{j-1}} \, .
  \end{equation}
  According to \eqref{unitaryconversion}, the relationship between the
  decompositions $\{P_i, \ket{\Psi_{i}} \}$ and $\{q_j,
  \ket{\Theta_{j}} \}$ is $\sqrt{P_i} \ket{\Psi_{i}} = \sum^{r}_{j=1}
  u_{ij} \sqrt{q_j} \ket{\Theta_j}$, where $u_{ij}$ is some unitary
  matrix.  Because of the form of the $\ket{\Theta_j}$, this
  immediately yields the Schmidt decomposition of $\ket{\Psi_i}$. The
  GM of pure bipartite states is determined by their largest Schmidt
  coefficient, so we have
  \begin{equation}\label{ea:Gcl>Gfl2}
    P_i \Lam (\ket{\Psi_{i}}) =
    \max_{ j \in \{1, \ldots ,r\} }
    \Bigg\{ \bigg| \frac{u_{ij}\sqrt{q_j}}{\sqrt{n_j - n_{j-1} } }
    \bigg|^2 \Bigg\}
    \leq
    \sum^r_{j=1} \frac{u_{ij}^{*} u_{ij} \, q_j}{n_j - n_{j-1}} \, .
  \end{equation}
  By summing over all $i$, we obtain $\sum_i P_i \Lam (\ket{\Psi_{i}})
  \leq \sum^{r}_{j=1} \frac{q_j}{n_j - n_{j-1}}$. Comparing this
  inequality to \eqref{ea:Gcl>Gfl}, we see that the inequalities must
  become equalities, and therefore $\{q_j, \ket{\Theta_j} \}$ is also
  optimal for $\Gc( \rho )$.  Since the inequality in
  \eqref{ea:Gcl>Gfl2} becomes an equality, all but one $u_{i1} ,
  \ldots , u_{ir}$ are zero. Hence, the state $\ket{\Psi_{i}}$ is
  identical to one of the states $\ket{\Theta_j}$, $j=1, \ldots, r$,
  up to an overall phase. Therefore, $\{q_i, \ket{\Theta_i} \}$ is the
  unique optimal decomposition for $\Gc( \rho )$, up to overall
  phases.

  With item 1 proved, the other items easily follow:
  \begin{itemize}
  \item Item 2 and 3: These follow directly from item 1 and
    $\Lam(\ket{\Theta_{i}}) = \frac{1}{n_i-n_{i-1}}$, together with
    Proposition \ref{prop:Gc=Gf} and \eqref{ea:convexineq},
    respectively.

  \item Item 4: If $\Gcl( \rho ) = \Gfl( \rho )$, then it follows from
    item 1 and Theorem \ref{thm:Gfl=Gcl} that the $\ket{\Theta_i}$ are
    all equally entangled. Conversely, if the $\ket{\Theta_i}$ are all
    equally entangled, then it follows from item 1 and Corollary
    \ref{coroptgcgcl} that $\Gcl( \rho ) = \Gfl( \rho )$.

  \item Item 5: Because of $\Lam(\ket{\Theta_{i}}) =
    \frac{1}{n_i-n_{i-1}}$, items 4 and 5 are equivalent.

  \item Item 6: If $\Gcl ( \rho ) = \Gfl ( \rho )$, then it follows
    from item 1 and Theorem \ref{thm:Gfl=Gcl} that every optimal
    decomposition of $\Gcl ( \rho )$ must be of the form $\{ q_i ,
    \ket{\Theta_i} \}$, up to overall phases.  Since overall phases do
    not change the value of \eqref{ea:convexroof1} or
    \eqref{ea:convexroof2}, $\Gc ( \rho )$ and $\Gcl ( \rho )$ have
    the same optimal decompositions.
    \qed
   \end{itemize}
\end{proof}
One may wonder whether the optimal decomposition for $\Gc(\rho)$ in
Lemma \ref{le:Gcl>Gfl} is also optimal for $\Gcl(\rho)$ when
$\Gcl(\rho) \neq \Gfl(\rho)$. In the following we show that this is
the case for all maximally correlated qutrit states. Only rank-2
states need to be considered, because for $d=3$ this is the only
non-trivial case.
\begin{proposition}\label{pp:twoqutritGcl}
  Let $q \in (0,1)$ and $\ket{\psi} = \frac{1}{\sqrt2} (\ket{11} +
  \ket{22} )$.  The maximally correlated two-qutrit state $\rho = q
  \proj{00} + (1-q) \proj{\psi}$ has $\Gcl(\rho) = 1-q$, with $\{ q,
  \ket{00}; 1-q, \ket{\psi} \}$ being an optimal decomposition.
\end{proposition}
\begin{proof}
  Let $\{P_i, \ket{\Psi_{i}}\}$ be some optimal decomposition of
  $\rho$ for $\Gcl ( \rho )$.  According to \eqref{unitaryconversion},
  there exists a unitary $u_{ij}$, so that $\sqrt{P_i} \ket{\Psi_{i}}
  = u_{i1} \sqrt{q} \ket{00}+u_{i2} \sqrt{1-q} \ket{\psi}$ for each
  $i$.  Setting $x_i := \abs{u_{i1}\sqrt{q}}^2$ and $y_i :=
  \abs{u_{i2}\sqrt{1-q}}^2$, we have $P_i = x_i + y_i$ and $\Lam
  (\ket{\Psi_{i}}) = \max \{ \frac{x_i}{P_i} , \frac{y_i}{2 P_i} \}$.
  Therefore,
  \begin{equation}
    \Gcl(\rho) = - \sum^{r}_{i=1} P_i \log \left[ \max
    \Big\{ \frac{x_i}{P_i} , \frac{y_i}{2 P_i} \Big\} \right] \, .
  \end{equation}
  Without loss of generality we assume $2 x_i \geq y_i$ for
  $i\in[1,k]$ and $2 x_i \leq y_i$ for $i \in [k+1, r]$.  We define $Y
  := \sum_{i=1}^{k} y_i \geq 0$ and $X := \sum_{i=1}^{k} x_i > 0$.
  Then
  \begin{equation*}
    \begin{split}
      \Gcl(\rho) &= \sum^k_{i=1} P_i \log \Big( \frac{P_i}{x_i} \Big) +
      \sum^r_{j=k+1} P_j \log \Big( \frac{2 P_j}{y_j} \Big) \\
      &=
      \sum^k_{i=1} (x_i + y_i) \log \Big( 1 + \frac{y_i}{x_i} \Big)
      +
      \sum^r_{j=k+1} (x_j + y_j) \log \Big(2 + \frac{2 x_j}{y_j} \Big) \\
      &\geq
      Y \Big[ ( 1 + \tfrac{X}{Y}) \log \big( 1 + \tfrac{Y}{X} \big) \Big]
      + \sum^r_{j=k+1} y_j
      \geq \frac{Y}{\ln 2} + \sum^r_{j=k+1} y_j
      \geq \sum_{i=1}^{r} y_i = 1-q \, .
    \end{split}
  \end{equation*}
  The first inequality follows by applying Lemma \ref{le:inequality}
  from Appendix \ref{ap:minimum} to the first sum, and using $x_j \geq
  0$, $y_j > 0$ in the second sum.  The application of Lemma
  \ref{le:inequality} is possible despite the restriction $2 x_i \geq
  y_i$, because the minimum in \eqref{ea:inequality} cannot be smaller
  with additional restrictions than without. The second inequality
  follows from $\inf_{x > 0} (1 + \frac{1}{x}) \log_2 (1+x) = 1/\ln2$.

  On the other hand, the decomposition $\{ q, \ket{00}; 1-q,
  \ket{\psi} \}$ yields $\Gcl(\rho) \leq 1-q$. This completes the
  proof.
  \qed
\end{proof}
One may conjecture that this proposition can be generalized to higher
dimensions, with $\ket{\psi} = \frac{1}{\sqrt n} \sum_{i=1}^{n}
\ket{ii}$ being a pure MES of any dimension, i.e.~that $\{ q ,
\ket{00} ; 1-q , \ket{\psi} \}$ is an optimal decomposition of $\rho =
q \proj{00} + (1-q) \proj{\psi}$ for $\Gcl ( \rho )$, yielding $\Gcl (
\rho ) = (1-q) \log n$.  From Proposition \ref{pp:twoquditGcl} it will
follow that this is not the case for $n>2$.  We will see that --
compared to the qutrit case -- the optimal decomposition for $\Gcl (
\rho )$ of higher-dimensional maximally correlated states is more
complex, even in the comparatively easy rank-2 case.  In the
following, $\E$ denotes the base of the natural logarithm.
\begin{proposition}\label{pp:twoquditGcl}
  Let $m, n \in \mathbb{N}$ with $\frac{m}{n} \leq 1$ and $q \in
  (0,1)$ be constants that define the rank-two maximally correlated
  qudit state $\rho = q \proj{\psi_m} + (1-q) \proj{\psi_n}$, with
  $\ket{\psi_m} = \frac{1}{\sqrt{m}} \sum^{m}_{i=1} \ket{ii}$ and
  $\ket{\psi_n} = \frac{1}{\sqrt{n}} \sum^{m+n}_{j=m+1} \ket{jj}$.
  Depending on the constants $m$, $n$ and $q$, an optimal
  decomposition and the value of $\Gcl (\rho)$ are
  \begin{itemize}
  \item $\frac{m}{n} \geq \frac{1}{\E}$: $\{ q , \ket{\psi_m} ; 1-q ,
    \ket{\psi_n} \}$, yielding $\Gcl(\rho) = q \log m + (1-q) \log n$.

  \item $\frac{m}{n} < \frac{1}{\E}$ and $q \geq \frac{\E m}{n}$: $\{
    \frac{1}{2}, \sqrt{q} \ket{\psi_m} \pm \sqrt{1-q} \ket{\psi_n}
    \}$, yielding $\Gcl(\rho) = \log (\frac{m}{q})$.

  \item $\frac{m}{n} < \frac{1}{\E}$ and $q < \frac{\E m}{n}$: $\{
    1-\frac{nq}{\E m} , \ket{\psi_n} ; \frac{nq}{2 \E m} ,
    \sqrt{\frac{\E m}{n}} \ket{\psi_m} \pm \sqrt{1-\frac{\E m}{n}}
    \ket{\psi_n} \}$, yielding $\Gcl(\rho) = \log n - q \tfrac{n \log
      \E}{m \E}$.
  \end{itemize}
\end{proposition}
\begin{proof}
  The above decompositions provide trivial upper bounds (note that in
  the second case $\Lam (\sqrt{q} \ket{\psi_m} \pm \sqrt{1-q}
  \ket{\psi_n}) = \frac{q}{m}$ follows from $\frac{q}{m} > \frac{q
    (1-q)}{m} \geq \frac{\E (1-q)}{n} > \frac{1-q}{n}$, and in the
  third case $\Lam (\sqrt{\frac{\E m}{n}} \ket{\psi_m} \pm \sqrt{1-
    \frac{\E m}{n}} \ket{\psi_n}) = \frac{\E}{n}$ follows from
  $\frac{e}{n} > \frac{1}{n} \geq \frac{1}{n} (1 - \frac{\E m}{n})$).
  Below we show that these are also lower bounds.

  Let $\{P_i, \ket{\Psi_{i}}\}$ be some optimal decomposition for
  $\Gcl( \rho )$.  According to \eqref{unitaryconversion}, there
  exists a unitary $u_{ij}$, so that $\sqrt{P_i} \ket{\Psi_{i}} =
  u_{i1}\sqrt{q} \ket{\psi_m} + u_{i2} \sqrt{1-q} \ket{\psi_n}$ for
  each $i$.  Setting $x_i := \abs{u_{i1}\sqrt{q}}^2$, and $y_i :=
  \abs{u_{i2}\sqrt{1-q}}^2$, we have $P_i = x_i + y_i$ and $\Lam
  (\ket{\Psi_{i}}) = \max \{ \frac{x_i}{m P_i} , \frac{y_i}{n P_i}
  \}$.  Therefore,
  \begin{equation*}
    \Gcl(\rho) = - \sum^{r}_{i=1} P_i \log \left[ \max
      \Big\{ \frac{x_i}{m P_i} , \frac{y_i}{n P_i} \Big\} \right] \, .
  \end{equation*}
  First, we rule out $n x_i = m y_i \:\: \forall i$ by showing that
  $\Gcl(\rho) = \sum_i P_i \log ( m P_i / x_i ) = \sum_i P_i \log ( m
  + n ) = \log (m + n)$ surpasses the upper bounds outlined above: For
  all $\frac{m}{n} \leq 1$ we have $\log (m + n) > \log n \geq q \log
  m + (1-q) \log n$, as well as $\log (m + n) > \log n \geq \log n - q
  \tfrac{n \log \E}{m \E}$.  Furthermore, for $\frac{m}{n} <
  \frac{1}{\E}$ and $q \geq \frac{\E m}{n}$ we have $\log (m + n) >
  \log n > \log (\frac{m n}{m \E}) \geq \log (\frac{m}{q})$.  In the
  following we therefore assume that $n x_i \neq m y_i$ holds for at
  least one $i \in [1,r]$.

  Without loss of generality we assume $n x_i \geq m y_i$ for
  $i\in[1,k]$ and $n x_i \leq m y_i$ for $i\in[k+1,r]$.  We define $Y
  := \sum_{i=1}^{k} y_i \geq 0$, $X := \sum_{i=1}^{k} x_i > 0$,
  $\widetilde{Y} := \sum_{i=k+1}^{r} y_i > 0$, and $\widetilde{X} :=
  \sum_{i=k+1}^{r} x_i \geq 0$, as well as $h := \frac{Y}{X} \geq 0$
  and $s := \frac{\widetilde{X}}{\widetilde{Y}} \geq 0$.  Note that $0
  \leq hs < 1$, because of $hs = \frac{Y}{X}
  \frac{\widetilde{X}}{\widetilde{Y}} < \frac{n}{m} \frac{m}{n} = 1$
  (the inequality is strict, because $n x_i > m y_i$ or $n x_i < m
  y_i$ holds for at least one $i$).  Using $Y + \widetilde{Y} = 1-q$
  and $X + \widetilde{X} = q$, it is easy to verify that $X = \frac{q
    - s(1-q)}{1 - hs}$ and $\widetilde{Y} = \frac{(1-q) - hq}{1 -
    hs}$.  From $X > 0$ and $\widetilde{Y} > 0$, it then follows that
  $s \in [0, \frac{q}{1-q}]$ and $h \in [0,\frac{1-q}{q}]$,
  respectively.
  \begin{equation}\label{eq:mN1}
    \begin{split}
      \Gcl(\rho) &=
      \sum^k_{i=1} P_i \log \Big( \frac{m P_i}{x_i} \Big) +
      \sum^r_{i=k+1} P_i \log \Big( \frac{n P_i}{y_i} \Big) \\
      &= \sum^{k}_{i=1} ( x_i + y_i )
      \log \Big[ m \Big( 1 +  \frac{y_i}{x_i} \Big) \Big] +
      \sum^r_{i=k+1} ( x_i + y_i )
      \log \Big[ n \Big( 1 + \frac{x_i}{y_i} \Big) \Big] \\
      &\geq X (1 + h)
      \log \big[ m ( 1 + h ) \big] +
      \widetilde{Y} ( 1 + s ) \log \big[ n ( 1 + s ) \big] \, ,
    \end{split}
  \end{equation}
  where the inequality follows from applying Lemma \ref{le:inequality}
  to each of the two sums.  In Lemma \ref{le:fhs} of Appendix
  \ref{ap:minimum} we show that the last line of \eqref{eq:mN1} is an
  upper bound to the values outlined in the proposition.  Hence, the
  upper and lower bounds coincide. This completes the proof.
  \qed
\end{proof}
Note that for qutrits ($m + n = 3$) Proposition \ref{pp:twoquditGcl}
simplifies to Proposition \ref{pp:twoqutritGcl}, because $\frac{1}{\E}
< \frac{1}{2} \leq \frac{m}{n}$.  From the symmetry of $\rho = q
\proj{\psi_m} + (1-q) \proj{\psi_n}$, it can be seen that Proposition
\ref{pp:twoquditGcl} can be extended to the case $\frac{m}{n} > 1$
simply by swapping $q$ and $1-q$.  Importantly, Proposition
\ref{pp:twoquditGcl} yields necessary and sufficient conditions for
rank-2 maximally correlated qudit states to have common optimal
decompositions for $\Gc (\rho)$ and $\Gcl (\rho)$:
\begin{corollary}\label{exclusiv}
  Rank-2 maximally correlated states $\rho = q \proj{\psi_m} + (1-q)
  \proj{\psi_n}$ have common optimal decompositions for $\Gc ( \rho )$
  and $\Gcl ( \rho )$ iff $\frac{1}{\E} \leq \frac{m}{n} \leq \E$.
\end{corollary}
\begin{proof}
  According to Lemma \ref{le:Gcl>Gfl}, the unique optimal
  decomposition for $\Gc (\rho)$ is $\{ q , \ket{\psi_m} ; 1-q ,
  \ket{\psi_n} \}$, up to overall phases.  For symmetry reasons it
  suffices to consider $\frac{m}{n} \leq 1$.  For $\frac{m}{n} \geq
  \frac{1}{\E}$ the statement immediately follows from Proposition
  \ref{pp:twoquditGcl}.  For $\frac{m}{n} < \frac{1}{\E}$ it is seen
  in the proof of Lemma \ref{le:fhs} that for all $q \in (0,1)$ the
  minimum of $f (h,s)$ is strictly smaller than $f (0,0) = q \log m +
  (1-q) \log n$. Therefore, $\{ q , \ket{\psi_m} ; 1-q , \ket{\psi_n}
  \}$ cannot be an optimal decomposition of $\Gcl (\rho)$ for
  $\frac{m}{n} < \frac{1}{\E}$.
  \qed
\end{proof}
Let us sum up the preceding findings.  Theorem \ref{thm:Gfl=Gcl} gives
necessary and sufficient conditions for $\Gcl(\rho)=\Gfl(\rho)$.
Apart from the trivial classes of pure states and separable states,
this includes all two-qubit states (Proposition \ref{pr:fctwoqubits}),
and all isotropic states (Proposition \ref{pr:isotropic}).  Further
examples are the 4-qubit Smolin state, and multiqubit D\"{u}r states,
for whom $\Gcl ( \rho ) = - \log_2 ( 1 - \Gc ( \rho ))$ can be easily
verified from \cite{wag04}.

According to Lemma \ref{le:Gcl>Gfl}, $\Gcl( \rho ) > \Gfl( \rho )$
holds for maximally correlated states iff there are two states
$\ket{\Theta_i}$ and $\ket{\Theta_j}$ in the optimal decomposition of
$\Gc ( \rho )$ that are not equally entangled. This is the case for
the two-qutrit states of Proposition \ref{pp:twoqutritGcl}. Despite
this, Proposition \ref{pp:twoqutritGcl} shows that $\Gc ( \rho )$ and
$\Gcl ( \rho )$ still have common optimal decompositions, i.e.~item
3(a) of Theorem \ref{thm:Gfl=Gcl} can be true, while item 1 is
false. In this case, the first inequality of \eqref{ea:Gfl=Gcl} is an
equality, while the second inequality is strict.

Optimal decompositions of two-qutrit isotropic states were found in
\cite{tv00} for the entanglement of formation, a convex roof-based
entanglement measure.  In all of these optimal decompositions
\emph{some} of the pure states are not equally entangled (although it
is unknown whether \cite{tv00} exhausts all optimal decompositions for
two-qutrit isotropic states).  Using Lemma \ref{le:Gcl>Gfl}, one can
easily construct a state $\rho$ in whose optimal decomposition $\{P_i,
\ket{\Psi_{i}}\}$ for $\Gc(\rho)$ \emph{any} two states $\ket{\Psi_i}$
and $\ket{\Psi_j}$ are not equally entangled. In some sense, this is a
stronger result than the one of \cite{tv00}.

Next, we address the question whether $\Gcl$ has the same ordering as
$\Gfc$ or $\Gfl$. Because of $\Gfc \cong \Gfl$, the statement $\Gcl
\cong \Gfc$ is equivalent to $\Gcl \cong \Gfl$. The two-qubit case is
trivial, because of $\Gcl \equiv \Gfl$ (cf.~Proposition
\ref{pr:fctwoqubits}), so we need to consider higher-dimensional
systems.
\begin{corollary}\label{ineqorder}
  In general, $\Gfc$ and $\Gcl$ do not have the same ordering ($\Gfc
  \ncong \Gcl$).  Equivalently, $\Gfl$ and $\Gcl$ do not have the same
  ordering ($\Gfl \ncong \Gcl$).
\end{corollary}
\begin{proof}
  A simple counterexample are the two maximally correlated six-level
  states $\rho = \frac{1}{2} \proj{\Psi_{123}} + \frac{1}{2}
  \proj{\Psi_{456}}$ and $\sigma = \frac{1}{3} \proj{\Psi_{12}} +
  \frac{2}{3} \proj{\Psi_{3456}}$, where $\ket{\Psi_{a \ldots z}} :=
  \frac{1}{\sqrt{z-a+1}} (\ket{aa} + \ldots + \ket{zz})$.  From items
  2 and 4 of Lemma \ref{le:Gcl>Gfl} it follows that $\Gfc ( \rho ) =
  \Gfc ( \sigma ) = \frac{2}{3}$, but $\Gcl ( \rho ) < \Gcl ( \sigma
  )$.
  \qed
\end{proof}
We remark that Corollary \ref{ineqorder} can be easily verified for a
much wider range of systems, e.g.~all bipartite $d$-level systems with
$d \geq 4$, by considering any rank-2 maximally correlated state
$\rho$ belonging to the second class outlined in Proposition
\ref{pp:twoquditGcl}, together with a suitably chosen isotropic state
$\sigma$, yielding $\Gfc ( \rho ) < \Gfc ( \sigma )$ and $\Gcl ( \rho
) > \Gcl ( \sigma )$.

\subsection{Comparison between $\Gm / \Gml$ and $\Gc / \Gcl$}

In contrast to the convex roof-based extensions, $\Gm$ and $\Gml$ are
demonstrably not convex, and they attain their maximum for the
maximally mixed state. It is therefore intuitive to expect that $\Gm (
\rho ) \geq \Gc ( \rho )$ and $\Gml ( \rho ) \geq \Gcl ( \rho )$ hold
for all $\rho$. To prove these statements, we need the following
lemma.
\begin{lemma}\label{le:GMequalPURE}
  Let $\rho , \ket{\varphi}$ be two arbitrary states, and
  $\bra{\varphi} \rho \ket{\varphi} = g$. Then, there exists a
  decomposition $\rho = \sum^r_{i=1} p_i \proj{\psi_i}$, such that $r=
  \rank \rho$, and $\abs{\braket{\varphi}{\psi_i}}^2 = g$ for all $i$.
\end{lemma}
\begin{proof}
  We use induction on the $\rank$ of $\rho$. The claim is trivial for
  $\rank \rho = 1$. Suppose, it is true for $\rank \rho = r$. Consider
  a general state $\rho$ with $\rank \rho = r+1$ and spectral
  decomposition $\rho = \sum^{r+1}_{i=1} p_i \proj{\psi_i}$, with
  $\braket{\psi_i}{\psi_j} = 0$ for $i \neq j$. Since $\bra{\varphi}
  \rho \ket{\varphi} = g$, we can assume
  $\abs{\braket{\varphi}{\psi_1}}^2 \leq g \leq
  \abs{\braket{\varphi}{\psi_2}}^2$ without loss of generality.
  Denote $\braket{\varphi}{\psi_j} = s_j \E^{\I \theta_j}$, $s_j \geq
  0$ for $j=1, 2$.  Using \eqref{unitaryconversion}, we rewrite the
  sum of the first two terms as
  \begin{equation*}
    p_1 \proj{\psi_1} + p_2 \proj{\psi_2} =
    q_1 \proj{\phi_1} + q_2 \proj{\phi_2} \, ,
  \end{equation*}
  with $\sqrt{q_i} \ket{\phi_{i}} = u_{i1} \sqrt{p_1} \ket{\psi_1} +
  u_{i2} \sqrt{p_2} \ket{\psi_2}$ for $i = 1, 2$, and where we define
  the unitary matrix as
  \begin{equation*}
    U ( \vartheta ) =
    [u_{ij}]=
    \begin{pmatrix}
      \co & \si \E^{\I (\theta_1 - \theta_2 )} \\
      - \si \E^{- \I (\theta_1 - \theta_2 )} & \co
    \end{pmatrix}
    \, , \quad \vartheta \in (0, \pi) \, ,
  \end{equation*}
  with $\co := \cos \frac{\vartheta}{2}$ and $\si := \sin
  \frac{\vartheta}{2}$. From $\braket{\psi_1}{\psi_2} = 0$ and
  $\braket{\phi_1}{\phi_1} = 1$, we obtain $q_1 = \co^2 p_1 + \si^2
  p_2$, hence
  \begin{equation*}
    \ket{\phi_{1}} = \frac{\co \sqrt{p_1} \ket{\psi_1}
      + \si \E^{\I (\theta_1 - \theta_2)} \sqrt{p_2}
      \ket{\psi_2}}{(\co^2 p_1 + \si^2 p_2)^{\frac{1}{2}}}
    \, , \: \text{and} \quad
    \abs{\braket{\varphi}{\phi_1}} =
    \frac{\co \sqrt{p_1} s_1 + \si \sqrt{p_2} s_2}{(\co^2 p_1 +
      \si^2 p_2)^{\frac{1}{2}}} \, .
  \end{equation*}
  We see that $\lim_{\vartheta \rightarrow 0}
  \abs{\braket{\varphi}{\phi_1}} = s_1$, and $\lim_{\vartheta
    \rightarrow \pi} \abs{\braket{\varphi}{\phi_1}} = s_2$.  Since
  $s_1^2 \leq g \leq s_2^2$ and $\abs{\braket{\varphi}{\phi_1}}$ is
  continuous in $\vartheta$, there must be some $\vartheta$ such that
  $\abs{\braket{\varphi}{\phi_1}}^2 = g$.  Denoting the corresponding
  $\{ q_i , \ket{\phi_i} \}$ as $\{ Q_i , \ket{\Phi_i} \}$, we see
  that the state
  \begin{equation*}
    \rho_1 := \frac{\rho - Q_1 \proj{\Phi_1}}{1 - Q_1} =
    \frac{Q_2 \proj{\Phi_2} + \sum^{r+1}_{i=3} p_i
      \proj{\psi_i}}{1 - Q_1}
  \end{equation*}
  is a state of rank $r$ that satisfies $\bra{\varphi} \rho_1
  \ket{\varphi} = g$. Using the induction assumption on $\rho_1$,
  there is a decomposition $\rho_1 = \sum^r_{i=1} p_i'
  \proj{\psi_i'}$, such that $\abs{\braket{\varphi}{\psi_i'}}^2 = g$
  for all $i$.  Now the claim follows for $\rho = (1 - Q_1) \rho_1 +
  Q_1 \proj{\Phi_1} = (1 - Q_1) \sum^r_{i=1} p_i' \proj{\psi_i'} + Q_1
  \proj{\Phi_1}$. This completes the proof.
  \qed
\end{proof}

\begin{theorem}\label{gmlgclinequality}
  $\Gml ( \rho ) \geq \Gcl ( \rho )$ holds for all states $\rho$.
\end{theorem}
\begin{proof}
  Let $\ket{\varphi} \in \pro$ be a closest product state of $\rho$ in
  accordance with \eqref{eq:EE}, and let $\rank \rho = r$. By virtue
  of Lemma \ref{le:GMequalPURE}, there exists a decomposition $\rho =
  \sum^r_{i=1} P_i \proj{\Psi_i}$ such that
  $\abs{\braket{\varphi}{\Psi_i}}^2 = \bra{\varphi} \rho
  \ket{\varphi}$ for all $i$. Then we have
  \begin{multline}\label{ea:Gml>=Gcl}
    \Gml ( \rho ) =
    - \log \bra{\varphi} \rho \ket{\varphi}
    = - \log \sum^r_{i=1} P_i
    \abs{\braket{\varphi}{\Psi_i}}^2
    =  - \sum^r_{i=1} P_i \log \abs{\braket{\varphi}{\Psi_i}}^2 \\
    \geq  - \sum^r_{i=1} P_i \log \Lam (\ket{\Psi_i})
    \geq \min\limits_{ \{ p_{i}, \ket{\psi_{i}} \} } - \sum\limits_i
    p_i \log \Lam (\ket{\psi_i})
    = \Gcl ( \rho ) \, ,
  \end{multline}
  where the third equality follows from the fact that the
  $\abs{\braket{\varphi}{\Psi_i}}$ have the same value for all $i$.
  \qed
\end{proof}

\begin{corollary}\label{gmgcinequality}
  $\Gm ( \rho ) \geq \Gfc ( \rho )$ holds for all states $\rho$.
\end{corollary}
\begin{proof}
  Using Lemma \ref{le:weakmonotonicity}, Theorem
  \ref{gmlgclinequality} and Lemma \ref{gcandgcl}, we obtain $\Gm (
  \rho ) = 1 - 2^{- \Gml ( \rho )} \geq 1 - 2^{- \Gcl ( \rho )} \geq
  \Gc ( \rho )$ .
  \qed
\end{proof}
The following Theorem \ref{thm:Gcl=Gml} establishes necessary and
sufficient conditions for $\Gml ( \rho ) = \Gcl ( \rho )$ in form of a
straightforward relationship between the CPS for $\Lamm ( \rho )$ and
the optimal decomposition of $\Gcl ( \rho )$. Hence, this theorem
bears resemblance to Theorem \ref{thm:Gfl=Gcl}, as well as to
Proposition 5 of \cite{skb11NJP}.
\begin{theorem}\label{thm:Gcl=Gml}
  For any state $\rho$ the following two conditions are equivalent:
  \begin{enumerate}
  \item $\Gml (\rho ) = \Gcl (\rho )$ holds.

  \item For every CPS $\phim$ of $\Lamm ( \rho )$ there exists an
    optimal decomposition $\{P_i, \ket{\Psi_{i}}\}$ of $\Gcl ( \rho )$
    for which $\Lamm ( \rho ) = \Lam (\ket{\Psi_i}) =
    \abs{\braket{\phi_{\text{m}}}{\Psi_i}}^2$ holds for all $i$.
  \end{enumerate}
\end{theorem}
\begin{proof}
  \begin{itemize}
  \item $1. \Rightarrow 2.$: Because of $\Gcl (\rho ) = \Gml (\rho )$,
    the two inequalities in \eqref{ea:Gml>=Gcl} must become
    equalities, from which it follows that for every CPS
    $\phim$ there exists a decomposition $\{P_i,
    \ket{\Psi_{i}}\}$ of $\rho$ which is optimal for $\Gcl ( \rho )$
    and for which $\Lam (\ket{\Psi_i}) =
    \abs{\braket{\phi_{\text{m}}}{\Psi_i}}^2 = \bra{\phi_{\text{m}}}
    \rho \ket{\phi_{\text{m}}} = \Lamm ( \rho )$ holds for all $i$.

  \item $2. \Rightarrow 1.$: Let $\phim$ be a CPS for $\Lamm ( \rho )$
    and let $\{P_i, \ket{\Psi_{i}}\}$ be an optimal decomposition of
    $\Gcl ( \rho )$ for which $\Lamm ( \rho ) = \Lam (\ket{\Psi_i}) =
    \abs{\braket{\phi_{\text{m}}}{\Psi_i}}^2$ holds for all $i$.
    Then, $\Gcl (\rho ) = - \sum_i P_i \log \Lamm ( \rho ) = - \log
    \Lamm ( \rho ) = \Gml ( \rho )$.
    \qed
\end{itemize}
\end{proof}
In Section \ref{partitioning} this theorem will be demonstrated by
means of the maximally correlated states.  With regard to the linear
measures, it will be shown in Theorem \ref{thm:equivalence} that $\Gm
( \rho ) = \Gfc ( \rho )$ holds iff $\rho$ is pure.

\subsection{Comparison between $\Gt$ and $\ET$}\label{relgtet}

\begin{theorem}\label{etgfc}
  $\Gfc ( \rho ) \geq \ETw ( \rho )$ holds for all states $\rho$.
\end{theorem}
\begin{proof}
  $\Gf ( \rho ) =
  1 - F^2 ( \rho , \sigma_{\text{f}} )
  \stackrel{\eqref{trfidineq}}{\geq}
  \DT^2 ( \rho , \sigma_{\text{f}} ) \geq
  \min\limits_{\sigma \in \sep} \DT^2 ( \rho , \sigma ) =
  \ETw ( \rho )$ .
  \qed
\end{proof}

\begin{theorem}\label{gtgm}
  $\Gt ( \rho ) \leq \Gm ( \rho )$ holds for all states $\rho$.
\end{theorem}
\begin{proof}
  Let $\rho = \sum_i p_i \proj{\psi_i}$ be an arbitrary decomposition
  of $\rho$. Then,
  \begin{multline*}
    \Gt ( \rho ) = \min_{\ket{\varphi} \in \pro}
    \DT^2 ( \rho , \ket{\varphi} ) \leq
    \min_{\ket{\varphi}\in \pro} \bigg[ \sum_i p_i
    \DT^2 ( \ket{\psi_i} , \ket{\varphi} ) \bigg] \\
    \stackrel{\eqref{trfidineq}}{=}
    \min_{\ket{\varphi}\in \pro} \bigg[ \sum_i p_i \left( 1 -
        \abs{\braket{\varphi}{\psi_i}}^2 \right) \bigg]
      = 1 - \max_{\ket{\varphi} \in \pro}
      \bra{\varphi} \rho \ket{\varphi} =
    \Gm ( \rho ) \, ,
  \end{multline*}
  where the inequality follows from the convexity of $\DT^2$.
  \qed
\end{proof}
The inequality $\Gt ( \rho ) \leq \Gm ( \rho )$ can be strict, as seen
for $\rho = \one / d^2$, the maximally mixed state of two qudits: $\Gt
( \rho ) = ( 1 - \frac{1}{d^2} )^2$ (cf.~Appendix
\ref{ap:gtpconcave}), which is always smaller than $\Gm ( \rho ) = 1 -
\frac{1}{d^2}$.

\subsection{Inequalities and hierarchies}\label{hierarchy}

Using the results from the preceding sections, we find the following
inequality chains that include all the GM definitions considered in
this paper.  These inequality hierarchies are summarized and
visualized in Figure \ref{GMhierarchy}.
\begin{theorem}\label{th:ineqalities}
  The following inequalities hold for all states $\rho$:
  \begin{enumerate}
  \item $\ETw ( \rho ) \leq \Gt ( \rho ) \leq \Gm ( \rho )$

  \item $\ETw ( \rho ) \leq \Gfc ( \rho ) \leq \Gm ( \rho )$

  \item $\Gm ( \rho ) \leq \Gml ( \rho ) \leq \ER ( \rho ) + S ( \rho
    )$

  \item $\Gfc ( \rho ) \leq \Gfl ( \rho ) \leq \Gcl ( \rho ) \leq \Gml
    ( \rho )$
  \end{enumerate}
\end{theorem}
\begin{proof}
  \begin{itemize}
  \item 1: The first inequality was shown in Section \ref{sub:gt}, and
    the second one in Theorem \ref{gtgm}.
    
  \item 2: The first inequality was shown in Theorem \ref{etgfc}, and
    the second one in Corollary \ref{gmgcinequality}.
    
  \item 3: The first inequality follows from Lemma
    \ref{le:weakmonotonicity}, and the second one was shown in
    \cite{hmm06,weg04}.
    
  \item 4: The first inequality follows from Lemma
    \ref{le:weakmonotonicity}, the second was shown in
    \eqref{ea:convexineq}, and the third in Theorem
    \ref{gmlgclinequality}.
    \qed
  \end{itemize}
\end{proof}

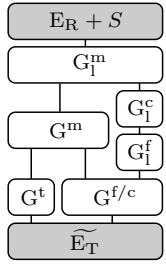
\begin{figure}
  \centering
  \begin{overpic}[scale=.14]{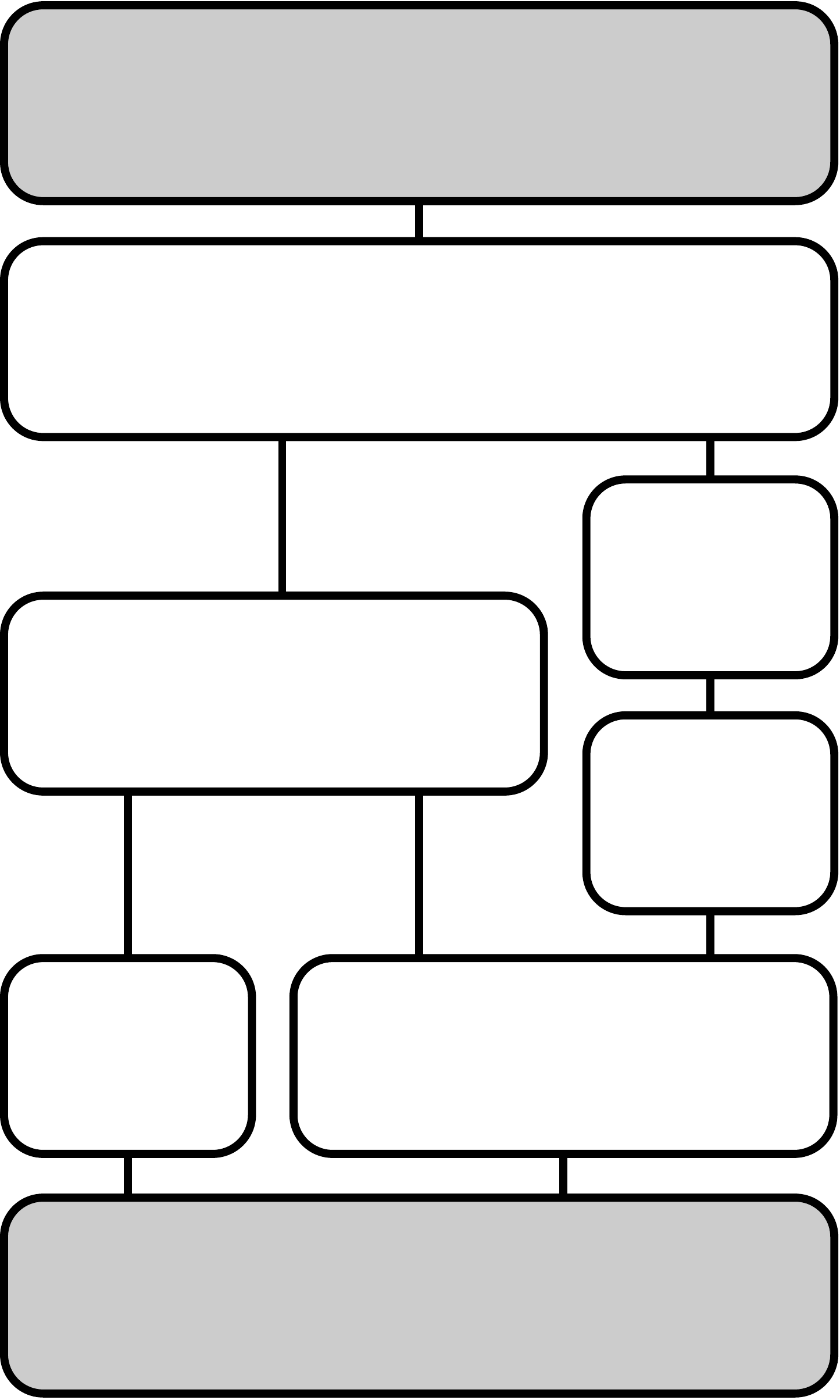}
    \put(16,90.5){$\ER + S$}
    \put(25,74){$\Gml$}
    \put(16,47.5){$\Gm$}
    \put(5,21.5){$\Gt$}
    \put(46,56.5){$\Gcl$}
    \put(46,39.5){$\Gfl$}
    \put(33,21.5){$\Gfc$}
    \put(24,04){$\ETw$}
  \end{overpic}
  \caption{\label{GMhierarchy} The quantitative hierarchy of the
    different measures is shown, with the six distinct extensions of
    GM in white boxes.  For a given state $\rho \in \SH$ the value of
    the measures increases monotonically from bottom to top along the
    vertical lines, and measures that are not vertically connected are
    not in an inequality relationship to each other.  The quantities
    $\ETw$ and $\ER + S$ are not extensions of GM, but they provide
    lower and upper bounds, respectively.}
\end{figure}

To verify that no inequality relationship exists for measures that are
not vertically connected in Figure \ref{GMhierarchy}, e.g.~$\Gt$ and
$\Gcl$, we need to find $\rho_1$, $\rho_2 \in \SH$ so that $\Gt
(\rho_1) < \Gcl(\rho_1)$, and $\Gt (\rho_2) > \Gcl(\rho_2)$. The
absence of an inequality relationship will be denoted as $\Gt \gtrless
\Gcl$.

\begin{proposition}\label{pr:incomparable}
  $\Gm \gtrless \Gcl$ , $\Gm \gtrless \Gfl$ , $\Gt \gtrless \Gcl$ ,
  $\Gt \gtrless \Gfl$ , and $\Gt \gtrless \Gfc$.
\end{proposition}
\begin{proof}
  Let $\rho_1$ be a genuinely mixed separable state and $\rho_2 =
  \proj{\psi}$ a pure entangled state.  Then, $\Gm ( \rho_1 ) > \Gcl (
  \rho_1 ) = \Gfl ( \rho_1 ) = 0$ and $\Gm ( \rho_2 ) = G( \ket{\psi}
  ) < \Gl( \ket{\psi} ) = \Gcl ( \rho_2 ) = \Gfl ( \rho_2 )$, from
  which it follows that $\Gm \gtrless \Gcl$ and $\Gm \gtrless \Gfl$.

  Equivalently, $\Gt ( \rho_1 ) > \Gcl ( \rho_1 ) = \Gfl ( \rho_1 ) =
  0$ and $\Gt ( \rho_2 ) = G( \ket{\psi} ) < \Gl( \ket{\psi} ) = \Gcl
  ( \rho_2 ) = \Gfl ( \rho_2 )$, from which it follows that $\Gt
  \gtrless \Gcl$ and $\Gt \gtrless \Gfl$.

  Let $\rho_1$ again be a genuinely mixed separable state, and $\rho_2
  = q \proj{00} + (1-q) \proj{\psi}$ the two-qutrit mixed entangled
  state of Proposition \ref{pp:twoqutritGcl}.  Then, $\Gt ( \rho_1 ) >
  \Gfc ( \rho_1 ) = 0$, and from item 2 of Lemma \ref{le:Gcl>Gfl} it
  follows that $\Gfc ( \rho_2 ) = \frac{1-q}{2}$. By choosing
  $\ket{\phi} = \ket{00}$ in \eqref{ea:Gtr}, we have $\Gt( \rho_2 )
  \leq (1-q)^2$. So, $\Gt ( \rho_2 ) < \Gfc ( \rho_2 )$ for $q \in
  (\frac{1}{2}, 1)$.  Hence, $\Gt \gtrless \Gfc$.
  \qed
\end{proof}

From Theorem \ref{th:ineqalities} and Figure \ref{GMhierarchy} we see
that $\Gm$ and $\Gml$ are upper bounds for all linear and logarithmic
GM definitions, respectively.  Since $\Lamm ( \rho )$ can be computed
for many prominent states (see e.g.~\cite{zch10}), these upper bounds
are readily accessible. For pure states the bounds are strict, and the
more mixed a given state $\rho$ is, the weaker the bounds are.

\subsection{Partitioning of state space}\label{partitioning}

As shown in the previous subsection, the inequalities
\begin{gather}
  0 \leq \Gfc ( \rho ) \leq \Gm ( \rho ) \, , \label{partineq1} \\
  0 \leq \Gfl ( \rho ) \leq \Gcl ( \rho ) \leq \Gml ( \rho )
  \, , \label{partineq2}
\end{gather}
hold for all states. Here we will see that these inequalities provide
a physically meaningful partitioning of state space. For this we first
prove that the measures in the above inequalities coincide iff $\rho$
is pure.
\begin{theorem}\label{thm:equivalence}
  For any state $\rho$ the following three conditions are equivalent:
  \begin{enumerate}
  \item $\Gfc (\rho ) = \Gm(\rho )$
    
  \item $\Gfl (\rho ) = \Gcl (\rho ) = \Gml (\rho )$
    
  \item $\rho$ is pure, i.e. $\rho = \proj{\psi}$
  \end{enumerate}
\end{theorem}

\begin{proof}
  \begin{itemize}
  \item 1 $\Rightarrow$ 2: If $\Gf (\rho ) = \Gm (\rho )$, then $\Lamf
    (\rho ) = \Lamm (\rho )$, and $\Gfl (\rho ) = \Gml (\rho )$. From
    \eqref{partineq2} it follows that $\Gfl ( \rho ) = \Gcl ( \rho ) =
    \Gml ( \rho )$.

  \item 2 $\Rightarrow$ 1: If $\Gfl (\rho ) = \Gml (\rho )$, then
    $\Lamf (\rho ) = \Lamm (\rho )$, and $\Gf (\rho ) = \Gm (\rho )$.

  \item 3 $\Rightarrow$ 2: Obvious from the definition of GM for pure
    states.

  \item 2 $\Rightarrow$ 3: Let $\{ p_i , \ket{\psi_i} \}$ be some
    decomposition of $\rho$, and let $\phim$ be a CPS of $\Lamm ( \rho
    )$.  Using the concavity of $\Lamf (\rho )$, we have
    \begin{equation}\label{lem17proof}
      \begin{split}
        \Lamf ( \rho ) &\geq \sum_i p_i \Lam (\ket{\psi_i}) =
        \sum_i p_i \max_{\ket{\varphi_i} \in \pro}
        \abs{\braket{\varphi_i}{\psi_i}}^2 \\
        &\geq
        \sum_i p_i
        \abs{\braket{\phi_{\text{m}}}{\psi_i}}^2 =
        \bra{\phi_{\text{m}}} \rho \phim
        = \max_{\ket{\varphi} \in \pro}
        \bra{\varphi} \rho \ket{\varphi} =
        \Lamm ( \rho ) \, .
      \end{split}
    \end{equation}
    From $\Gfl (\rho) = \Gml (\rho)$ it follows that $\Lamf (\rho) =
    \Lamm (\rho)$, and therefore the inequalities in
    \eqref{lem17proof} must become equalities.  From the second
    inequality of \eqref{lem17proof} it follows that there exists a
    $\phim \in \pro$ that is a CPS for all $\ket{\psi_i}$.
    Furthermore, since the choice of decomposition was arbitrary, this
    $\phim$ must be a CPS for all $\ket{\psi_i}$ of all conceivable
    decompositions $\{ p_i , \ket{\psi_i} \}$ of $\rho$. As shown
    below, this is possible only for pure states, i.e.~$\rank \rho =
    1$.

    Assume $r = \rank \rho \geq 2$. Because of $\Gfl (\rho) = \Gcl
    (\rho)$, it follows from Theorem \ref{thm:Gfl=Gcl} that there
    exists a decomposition $\{ P_i , \ket{\Psi_i} \}$ of $\rho$ where
    the $\ket{\Psi_i}$ are all equally entangled. Since $\phim$
    is a common CPS, we can write $\ket{\Psi_1}$ and $\ket{\Psi_2}$ as
    \begin{align*}
      \ket{\Psi_1} &= \alpha \phim +
      \sqrt{1- \alpha^2} \ket{\phi^{\perp}_{1}} \, , \\
      \ket{\Psi_2} &= \alpha \phim +
      \sqrt{1- \alpha^2} \ket{\phi^{\perp}_{2}} \, ,
    \end{align*}
    with $\alpha \in (0,1)$. Here, $\ket{\phi^{\perp}_{1}}$ and
    $\ket{\phi^{\perp}_{2}}$ are states orthogonal to $\phim$, and
    because of $\rank \rho \geq 2$, we can assume that
    $\ket{\phi^{\perp}_{1}} - \ket{\phi^{\perp}_{2}}$ is a non-zero
    vector. Clearly, $\Lam (\ket{\Psi_1}) = \Lam (\ket{\Psi_2}) =
    \alpha^2 > 0$.  Starting with $\{ P_i , \ket{\Psi_i} \}$, we use
    \eqref{unitaryconversion} to construct a new decomposition $\{ q_i
    , \ket{\Theta_i} \}$ of $\rho$ by using the $r \times r$ unitary matrix
    \begin{equation*}
      U =
      \begin{pmatrix}
        \sqrt{\frac{P_2}{P_1 + P_2}} & - \sqrt{\frac{P_1}{P_1 + P_2}}
        & 0 \\
        \sqrt{\frac{P_1}{P_1 + P_2}} & \sqrt{\frac{P_2}{P_1 + P_2}}
        & 0 \\
        0 & 0 & \one_{r-2}
      \end{pmatrix}
      \, ,
    \end{equation*}
    where $\one_{r-2}$ denotes the $(r-2)$-dimensional unit matrix.
    With this we have
    \begin{equation*}
      \ket{\Theta_1} \propto
      \sqrt{P_2} \sqrt{P_1} \ket{\Psi_1} -
      \sqrt{P_1} \sqrt{P_2} \ket{\Psi_2}
      = \sqrt{P_1 P_2} \sqrt{1- \alpha^2}
      \big( \ket{\Phi^{\perp}_1} - \ket{\Phi^{\perp}_2} \big) \, ,
    \end{equation*}
    and therefore, $\braket{\phi_{\text{m}}}{\Theta_1} = 0$. This is a
    contradiction to the requirement that $\phim$ is also a CPS for
    $\ket{\Theta_1}$. This completes the proof.
    \qed
  \end{itemize}
\end{proof}
According to Theorem \ref{thm:equivalence}, $\Gm$ (or $\Gml$)
coincides with the proper entanglement measure $\Gfc$ (or $\Gfl$) only
for pure states, thus further reinforcing the observation that $\Lamm
( \rho )$ assesses the entanglement as well as the mixedness of a
state $\rho$.  Together with the known fact that $\Lamf ( \rho ) < 1$
iff $\rho$ is entangled, we can use \eqref{partineq1} to partition the
state space into four subsets, $\SH = A \cup B \cup C \cup D$,
corresponding to pure separable, pure entangled, mixed separable and
mixed entangled states, respectively.  As shown in Table
\ref{tab:GMinequalities} and Figure \ref{statespace}, this
partitioning is done by determining whether the inequalities in
\eqref{partineq1} are strict or become equalities.  The inequalities
\eqref{partineq1} between the logarithmic measures can also be used
for partitioning $\SH$, and in that case, the subset of mixed
entangled states is further divided into three subsets, $D = D_1 \cup
D_2 \cup D_3$, because $\Gfc ( \rho ) < \Gm ( \rho )$ corresponds to
the three possible cases $\Gfl ( \rho ) < \Gcl ( \rho ) < \Gml ( \rho
)$, $\Gfl ( \rho ) = \Gcl ( \rho ) < \Gml ( \rho )$ and $\Gfl ( \rho )
< \Gcl ( \rho ) = \Gml ( \rho )$.

\begin{table}
  \caption{\label{tab:GMinequalities} The inequalities
    \protect\eqref{partineq1} facilitate a partitioning of state
    space $\SH$ into four subsets with clear physical meaning,
    and the inequalities \eqref{partineq2} allow for a further
    division of the genuinely mixed entangled states
    $D = D_1 \cup D_2 \cup D_3$ into three subsets.
    Theorem~\ref{thm:Gfl=Gcl} and Theorem~\ref{thm:Gcl=Gml} describe
    necessary and sufficient conditions for states belonging to $D_2$ and $D_3$, respectively.}
  \centering
  \begin{tabular}{|cc|ccc|}
    \hline
    \hline
    \multicolumn{2}{|c}{\multirow{2}{*}{Subset}} &
    Characterization &
    \multirow{2}{*}{Linear measures} &
    \multirow{2}{*}{Logarithmic measures} \\
    \multicolumn{2}{|c}{} &
    of states in subset & & \\
    \hline
    \multicolumn{2}{|c|}{$A$} &
    pure separable & $0 = \Gfc = \Gm$ & $0 = \Gfl = \Gcl = \Gml$ \\
    \hline
    \multicolumn{2}{|c|}{$B$} &
    pure entangled & $0 < \Gfc = \Gm$ & $0 < \Gfl = \Gcl = \Gml$ \\
    \hline
    \multicolumn{2}{|c|}{$C$} &
    genuinely mixed separable & $0 = \Gfc < \Gm$ & $0 = \Gfl = \Gcl < \Gml$ \\
    \hline
    \multicolumn{1}{|c|}{\multirow{3}{*}{$D$}} & $D_1$ &
    \multirow{3}{*}{genuinely mixed entangled} &
    \multirow{3}{*}{$0 < \Gfc < \Gm$} &
    \multicolumn{1}{|c|}{$0 < \Gfl < \Gcl < \Gml$} \\ \cline{2-2} \cline{5-5}
    \multicolumn{1}{|c|}{} & $D_2$ & & &
    \multicolumn{1}{|c|}{$0 < \Gfl = \Gcl < \Gml$} \\ \cline{2-2} \cline{5-5}
    \multicolumn{1}{|c|}{} & $D_3$ & & &
    \multicolumn{1}{|c|}{$0 < \Gfl < \Gcl = \Gml$} \\
    \hline
    \hline
  \end{tabular}
\end{table}

Since $\Gfl ( \rho ) = \Gcl ( \rho )$ holds for isotropic states and
all 2 qubit states, these states belong to the set $D_2$.  In
particular, for the special case of 2 qubits ($\cH = \bC^2 \otimes
\bC^2$) we have $D = D_2$, i.e. $D_1$ and $D_3$ are empty.  From this
one could conjecture that generic mixed entangled states belong to
$D_2$.  However, for maximally correlated states, it is clear from
item 5 of Lemma \ref{le:Gcl>Gfl} that most states do not belong to
$D_2$.  The following theorem elucidates the relationship between the
parameters of rank-two maximally correlated states and the subgroups
$D_1$, $D_2$ and $D_3$.  For this we note that $\Lamm ( \rho )$ can be
easily calculated for maximally correlated states of the form in
Proposition \ref{pp:twoquditGcl} as $\Lamm ( \rho ) = \max \{
\frac{q}{m} , \frac{1-q}{n} \}$.
\begin{theorem}\label{th:Dsubset}
  Depending on the value of the parameters $m,n \in \mathbb{N}$ with
  $\frac{m}{n} \leq 1$ and $q \in (0,1)$, the rank-two maximally
  correlated states of Proposition \ref{pp:twoquditGcl} belong to
  either of the three subsets of genuinely mixed entangled states:
  \begin{itemize}
  \item $D_2$, i.e.~$\Gfl ( \rho ) = \Gcl ( \rho ) < \Gml ( \rho )$:
    for $\frac{m}{n} = 1 \, ,$

  \item $D_3$, i.e.~$\Gfl ( \rho ) < \Gcl ( \rho ) = \Gml ( \rho )$:
    for $\frac{m}{n} < \frac{1}{\E}$ and $q \geq \frac{\E m}{n} \, ,$

  \item $D_1$, i.e.~$\Gfl ( \rho ) < \Gcl ( \rho ) < \Gml ( \rho )$:
    for all other parameter values.
  \end{itemize}
\end{theorem}
\begin{proof}
  From item 5 of Lemma \ref{le:Gcl>Gfl} it follows that $\Gfl ( \rho )
  = \Gcl ( \rho )$ iff $m = n$.  Hence, states belong to $D_2$ iff
  $\frac{m}{n} = 1$.  For $\frac{m}{n} < 1$ one can distinguish
  between $\rho \in D_1$ and $\rho \in D_3$ by determining whether
  $\Gcl ( \rho ) \leq \Gml ( \rho )$ is strict or not.

  Recall that $\Gml ( \rho ) = \min \{ \log ( \frac{m}{q} ) , \log (
  \frac{n}{1-q} ) \}$.  If $q \geq \frac{\E m}{n}$, then $\frac{m}{q}
  \leq \frac{n}{\E} < \frac{n}{1-q}$, yielding $\Gml ( \rho ) = \log (
  \frac{m}{q} )$.  Therefore, if $\frac{m}{n} < \frac{1}{\E}$ and $q
  \geq \frac{\E m}{n}$ (second case of Proposition
  \ref{pp:twoquditGcl}), then $\Gcl ( \rho ) = \Gml ( \rho ) = \log (
  \frac{m}{q} )$, yielding $\rho \in D_3$.

  Regarding the first and third case of Proposition
  \ref{pp:twoquditGcl}, it is seen from its proof (including Lemma
  \ref{le:fhs} and its proof) that for parameter values in the
  interior of the domain (i.e.~excluding $\frac{m}{n} = 1$) the value
  of $\Gcl ( \rho )$ is strictly smaller than both $\log ( \frac{m}{q}
  )$ and $\log ( \frac{n}{1-q} )$.  Therefore, $\Gcl ( \rho ) < \Gml (
  \rho ) = \min \{ \log ( \frac{m}{q} ) , \log ( \frac{n}{1-q} ) \}$
  holds, so $\rho \in D_1$.
  \qed
\end{proof}

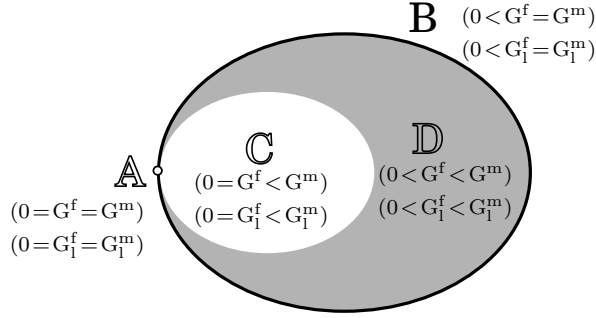
\begin{figure}
  \centering
  \begin{overpic}[scale=.3]{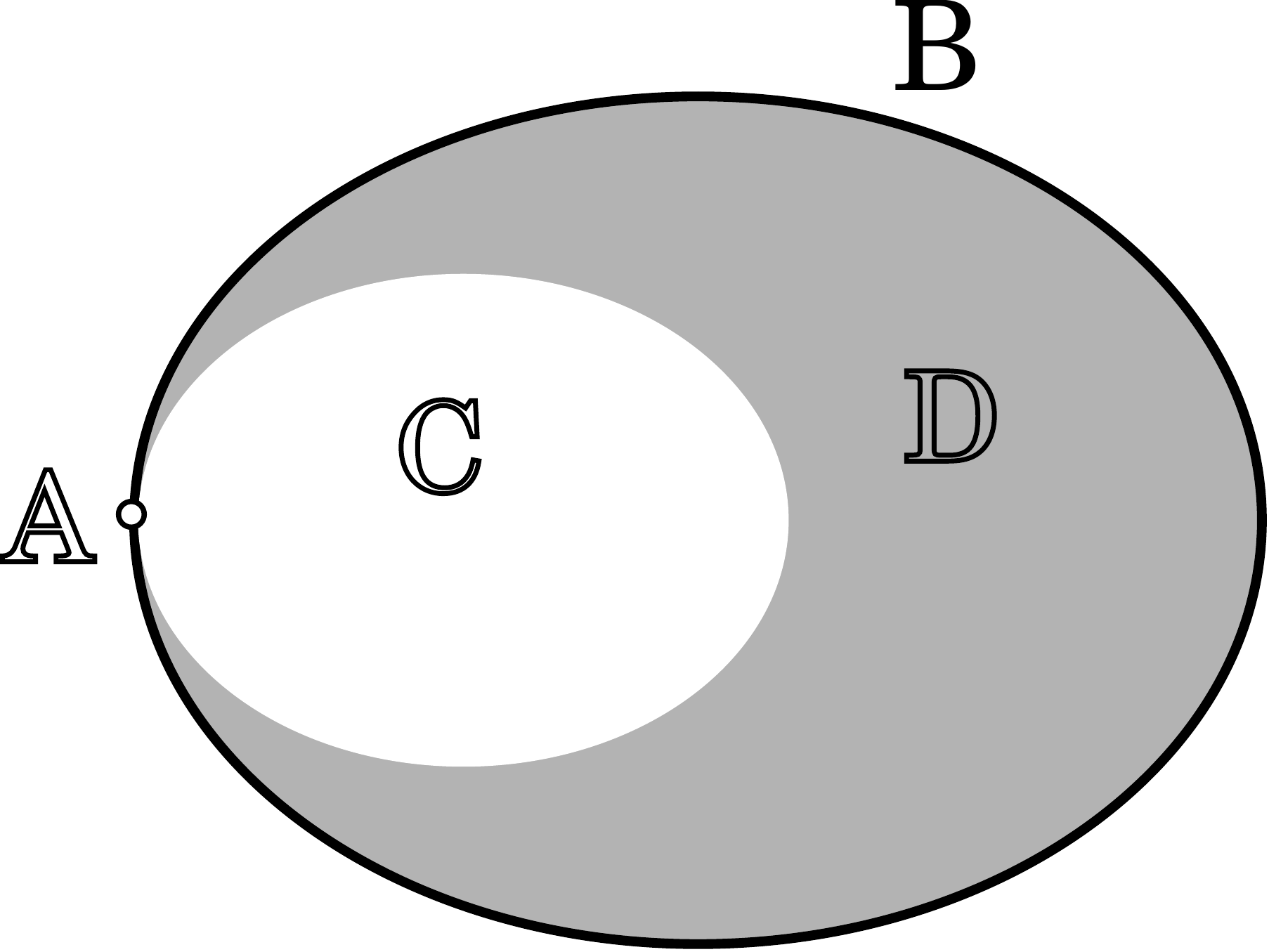}
    \put(-25,23){($0 \! = \! \Gf \! = \! \Gm$)}
    \put(-25,15){($0 \! = \! \Gfl \! = \! \Gml$)}
    \put(83,70){($0 \! < \! \Gf \! = \! \Gm$)}
    \put(83,62){($0 \! < \! \Gfl \! = \! \Gml$)}
    \put(63,32){($0 \! < \! \Gf \! < \! \Gm$)}
    \put(63,24){($0 \! < \! \Gfl \! < \! \Gml$)}
    \put(19,30){($0 \! = \! \Gf \! < \! \Gm$)}
    \put(19,22){($0 \! = \! \Gfl \! < \! \Gml$)}
  \end{overpic}
  \caption{\label{statespace} Illustration of the partitioning of
    state space $\SH = A \cup B \cup C \cup D$ into four pairwise
    disjoint sets. The pure states ($A \cup B$) lie at the boundary,
    and the separable states ($A \cup C$) form a closed, convex subset
    of $\SH$. The set a state $\rho$ belongs to is uniquely determined
    by whether the inequalities in \protect\eqref{partineq1} are
    strict or equalities. By means of $\Lamf (\rho)$ and $\Lamm
    (\rho)$, the same partitioning can be facilitated by the
    logarithmic quantities $\Gfl$ and $\Gml$. When taking $\Gcl$ into
    account, as well, the set of genuinely mixed entangled states $D =
    D_1 \cup D_2 \cup D_3$ is further subdivided into three pairwise
    disjoint sets.}
\end{figure}

From Theorem \ref{th:Dsubset} we see that the states belonging to
$D_3$ precisely coincide with those outlined in the second case of
Proposition \ref{pp:twoquditGcl}.  This allows us to demonstrate
Theorem \ref{thm:Gcl=Gml}: Let $\rho \in D_3$ be a rank-two maximally
correlated state. Then every CPS $\phim$ of $\Lamm ( \rho )$
necessarily has the form $\phim = ( \sum_{i=1}^{m} q_i \ket{i} )
\otimes ( \sum_{i=1}^{m} q_i^{*} \ket{i} )$ with $\sum_i \abs{q_i}^2 =
1$.  According to Proposition~\ref{pp:twoquditGcl}, an optimal
decomposition for $\Gcl(\rho)$ is $\rho = \frac{1}{2} \proj{\Psi_{+}}
+ \frac{1}{2} \proj{\Psi_{-}}$ with $\ket{\Psi_{\pm}} = \sqrt{q}
\ket{\psi_m} \pm \sqrt{1-q} \ket{\psi_n}$.  Using this decomposition,
we obtain $\Lamm ( \rho ) = \Lam (\ket{\Psi_{\pm}}) =
\abs{\braket{\phi_{\text{m}}}{\Psi_{\pm}}}^2 = \frac{q}{m}$, thus
verifying Theorem \ref{thm:Gcl=Gml}.

For maximally correlated qutrit states the only possible value of
$\frac{m}{n}$ is $\frac{1}{2}$, so all states lie in $D_1$. In
contrast to this, for four levels ($d = 4$), there are rank-two
maximally correlated states in each of the three sets $D_1$, $D_2$ and
$D_3$.

In the following we determine whether the various subsets of the
genuinely mixed states, $C \cup D = C \cup D_1 \cup D_2 \cup D_3$, are
convex or not.  $C$ and $C \cup D$ are clearly convex sets, while $D$
is not. All other subsets are investigated in the following lemma.
\begin{proposition}\label{nonconvex}
  The following sets are not convex: $D_1$, $D_2$, $D_3$, $D_1 \cup
  D_2$, $D_1 \cup D_3$, $D_2 \cup D_3$, $C \cup D_1$, $C \cup D_2$, $C
  \cup D_3$, and $C \cup D_1 \cup D_3$. Regarding $C \cup D_1 \cup
  D_2$ and $C \cup D_2 \cup D_3$, at least one of the two sets is not
  convex.
\end{proposition}
\begin{proof}
  To prove that $D_1$, $D_1 \cup D_3$, $C \cup D_1$ and $C \cup D_1
  \cup D_3$ are not convex, it suffices to find $\rho_1, \rho_2 \in
  D_1$ so that $\frac{1}{2} (\rho_1 + \rho_2) \in D_2$.  Using the
  notation $\ket{\Psi_{a \ldots z}} := \frac{1}{\sqrt{z-a+1}}
  (\ket{aa} + \ldots + \ket{zz})$ for MES, we choose $\rho_1 , \rho_2$
  to be
  \begin{equation*}
    \rho_{\pm} = \frac{1}{2} \proj{\Psi_{12}} +
    \frac{1}{2} \proj{\Psi_{3456}^{\pm}}
    \, , \quad \text{with} \quad
    \ket{\Psi_{3456}^{\pm}} := \frac{1}{\sqrt{2}}
    \big( \ket{\Psi_{34}} \pm \ket{\Psi_{56}} \big) \, .
  \end{equation*}
  Evidently, $\rho_{+}$ and $\rho_{-}$ are LU-equivalent, so they lie
  in the same set $D_i$ ($i = 1,2,3$).  Since $\frac{m}{n} =
  \frac{1}{2}$ for $\rho_{+}$, it follows from Theorem
  \ref{th:Dsubset} that $\rho_{\pm} \in D_1$.  On the other hand,
  \begin{equation*}
    \rho = \tfrac{1}{2} (\rho_{+} + \rho_{-}) =
    \frac{1}{2} \proj{\Psi_{12}} + \frac{1}{4} \proj{\Psi_{34}} +
    \frac{1}{4} \proj{\Psi_{56}} \, ,
  \end{equation*}
  so it follows from item 5 of Lemma \ref{le:Gcl>Gfl} that $\rho \in
  D_2$.

  Next, we prove that $D_3$ and $C \cup D_3$ are not convex by finding
  $\sigma_1, \sigma_2, \sigma_3 \in D_3$ that yield $\sigma =
  \frac{1}{3} (\sigma_1 + \sigma_2 + \sigma_3 ) \in D_2$. For this,
  consider
  \begin{gather*}
    \sigma_{i} = q \proj{\Psi_{12}} +
    (1-q) \proj{\Psi_{345678}^{i}} \, ,
    \quad \text{with }  \tfrac{\E}{3} \leq q < 1 \, , \\
    \text{and} \quad
    \ket{\Psi_{345678}^{i}} := \frac{1}{\sqrt{3}}
    \big( \ket{\Psi_{34}} + \E^{\I \frac{2 \pi i}{3}} \ket{\Psi_{56}} +
    \E^{\I \frac{4 \pi i}{3}} \ket{\Psi_{78}} \big) \, ,
  \end{gather*}
  for $i = 1,2,3$.  Evidently, the $\sigma_i$ are LU-equivalent, and
  since $\frac{m}{n} = \frac{1}{3}$ and $q \geq \frac{\E}{3}$, it follows
  from Theorem \ref{th:Dsubset} that $\sigma_i \in D_3$.  On the other
  hand,
  \begin{equation*}
    \sigma = \frac{1}{3} (\sigma_1 + \sigma_2 + \sigma_3 ) =
    q \proj{\Psi_{12}} + \frac{1-q}{3} \big( \proj{\Psi_{34}} +
    \proj{\Psi_{56}} + \proj{\Psi_{78}} \big) \, ,
  \end{equation*}
  so it follows from item 5 of Lemma \ref{le:Gcl>Gfl} that $\sigma \in
  D_2$.

  To prove that $D_2$, $D_1 \cup D_2$ and $D_2 \cup D_3$ are not
  convex, we consider $\rho_{\pm} = \frac{1}{2} \proj{\Psi_{12}^{\pm}}
  + \frac{1}{2} \proj{\Psi_{34}^{\pm}} \in D_2$, which yields $\rho =
  \frac{1}{2} ( \rho_{+} + \rho_{-}) = \frac{1}{4} \one \in C$.

  To prove that $C \cup D_2$ is not convex, we consider $\rho_{\pm} =
  \frac{1}{2} \proj{\Psi_{12}^{\pm}} + \frac{1}{2} \proj{\Psi_{34}}
  \in D_2$, and the genuinely mixed entangled state $\rho =
  \frac{1}{2} ( \rho_{+} + \rho_{-}) = \frac{1}{4} \proj{11} +
  \frac{1}{4} \proj{22} + \frac{1}{2} \proj{\Psi_{34}} \in D$. From
  item 5 of Lemma \ref{le:Gcl>Gfl} it follows that $\rho \notin D_2$,
  hence $\rho \in D_1 \cup D_3$. Although we do not know whether $\rho
  \in D_1$ or $\rho \in D_3$, we can ascertain that no more than one
  of the two sets $C \cup D_1 \cup D_2$ and $C \cup D_2 \cup D_3$ can
  be convex.
  \qed
\end{proof}
Although partially answered by Proposition \ref{nonconvex}, it is
still unknown whether $C \cup D_1 \cup D_2$ or $C \cup D_2 \cup D_3$
are convex or not.

\section{Graph states and cluster states}\label{graphstates}

Graph states are an important class of states for quantum information
\cite{heb04}. A subset of them, the cluster states, are the central
ingredient for one-way quantum computation \cite{br01}.  Here we show
that a large class of graph states, including all cluster states, have
a ``universal'' closest separable state that minimizes several
inequivalent distance measures.  This property helps to prove the
previous Corollary \ref{cor:traceent}.

Consider a general pure graph state $\ket{G}$ with underlying graph $G
= (V, E)$, where $V$ is the set of vertices and $E$ is the set of
edges. The \emph{maximum independent set} $\alpha$ is the largest
possible set of non-adjacent vertices, and the \emph{minimum vertex
  cover} $\beta$ is the complement of $\alpha$, i.e. $\alpha + \beta =
V$.  The minimum vertex cover can be thought of as the minimal set of
qubits that needs to be measured in the computational basis to
completely disentangle the graph state.

As outlined in \cite{hm13}, the stabilizer $\cS$ of $\ket{G}$ is
generated by $n$ generators $\{ g_j \}_{j=1}^{n}$, and these
generators stabilize a unique state, namely $\ket{G}$.  If one or more
of the generators from the generating set of $\cS$ are discarded, the
smaller set generates a new Abelian group $\cS'$ which now stabilizes
a set of states $\{ \ket{\psi_i} \}$ rather than a unique $\ket{G}$.
Depending on the structure of the generating set of $\cS'$, the states
$\{ \ket{\psi_i} \}$ may or may not be entangled. In \cite{hm13} it is
shown that the optimal way of discarding generators, such that the
stabilized states $\{ \ket{\psi_i} \}$ are product states, is to
discard generators corresponding to the vertices of the minimum vertex
cover. So, if we only keep the generators corresponding to the maximum
independent set $\{ g_j \vert j \in \alpha \}$, denoting the
correspondingly generated Abelian group as $\cS_{\alpha}$, the states
it stabilizes $\{ \ket{\psi_i^{\alpha}} \}$ are all product
states. These states form the basis vectors used below.

Ignoring possible negative amplitudes that are not important here,
$\ket{G}$ can be written as an equal superposition of the basis
vectors,
\begin{equation}\label{graph_form}
  \ket{G} = \frac{1}{\sqrt{D_{\alpha}}} \sum_{i=1}^{D_{\alpha}}
  \ket{\psi_i^{\alpha}} \, ,
\end{equation}
where $D_{\alpha}$ is the number of states $\ket{\psi_i^{\alpha}}$ and
is related to the cardinality of the minimum vertex cover as
$D_{\alpha} = 2^{\abs{\beta}}$. In other words, for each generator
discarded from the generating set, the size of the set of stabilized
states doubles.  The decomposition \eqref{graph_form} is of minimal
rank for graph states whose underlying graphs satisfy certain
conditions \cite{hm13}. This is the case for all bipartite
($2$-colorable) graphs, which includes all cluster states of arbitrary
size and dimension.  However, there also exist many non-$2$-colorable
states that satisfy the conditions.

In the following we assume that $\ket{G}$ satisfies the conditions,
i.e.~the decomposition \eqref{graph_form} is of minimal rank.  Since
the $\ket{\psi_i^{\alpha}}$ are product states, it immediately follows
from \eqref{graph_form} that any of the $\ket{\psi_i^{\alpha}}$ is a
CPS, i.e.~$\Lam ( \ket{G} ) = \abs{\braket{G}{\psi_i^{\alpha}}}^2 =
D_{\alpha}^{-1} = 2^{- \abs{\beta}}$.  Correspondingly, the separable
state
\begin{equation}\label{graph_css}
  \delta = \frac{1}{D_{\alpha}} \sum_{i=1}^{D_{\alpha}}
  \proj{\psi_i^{\alpha}} \, ,
\end{equation}
was found to be a CSS for the REE \cite{hm13}.  Here we show that it
is also the CSS in terms of the Bures distance \eqref{bures} and the
trace distance \eqref{trdist}.

\begin{theorem}\label{th:traceent}
  Let $\ket{G}$ be a graph state of the form \eqref{graph_form} with
  minimal rank. Then, \eqref{graph_css} is a closest separable state
  with respect to the quantum relative entropy, the Bures distance,
  and the trace distance.
\end{theorem}
\begin{proof}
  For the quantum relative entropy this was shown in \cite{hv10}, and
  we also know that $\Lam ( \ket{G} ) = D_{\alpha}^{-1}$.  Hence,
  $\max_{\sigma \in \sep} F^2 ( \ket{G} , \sigma ) = \max_{\ket{\phi}
    \in \pro} \abs{\braket{G}{\phi}}^2 = D_{\alpha}^{-1}$.  From $F^2
  ( \ket{G} , \delta ) = \abs{\bra{G} \delta \ket{G}} =
  D_{\alpha}^{-1}$ it then follows that $\delta$ minimizes the Bures
  distance \eqref{bures}.

  Next, consider the trace distance.  The inequality $\min_{\sigma \in
    \sep} D_{\text{T}} ( \ket{G} , \sigma ) \geq 1 - \max_{\sigma \in
    \sep} F^2 ( \ket{G} , \sigma ) = 1 - D_{\alpha}^{-1}$ follows from
  \eqref{tf2}.  To show that $\delta$ minimizes the trace distance, it
  therefore suffices to show that $D_{\text{T}} ( \ket{G} , \delta ) =
  1 - D_{\alpha}^{-1}$.  From \eqref{trdist} we know that
  $D_{\text{T}} ( \ket{G} , \delta ) = \frac{1}{2}
  \sum_{i=1}^{D_{\alpha}} \abs{\lambda_i}$, where the $\lambda_i$ are
  the eigenvalues of $A := \proj{G} - \delta = D_{\alpha}^{-1} \sum_{i
    \neq j} \ketbra{\psi_i^{\alpha}}{\psi_j^{\alpha}}$.  Taking the
  states $\{ \ket{\psi_i^{\alpha}} \}$ as the basis, elementary linear
  algebra yields the non-zero eigenvalues of $A$ as $\lambda_1 =
  \frac{D_{\alpha} - 1}{D_{\alpha}}$ and $\lambda_2 = \ldots =
  \lambda_{D_{\alpha}} = - \frac{1}{D_{\alpha}}$, hence $D_{\text{T}}
  ( \ket{G} , \delta ) = \frac{D_{\alpha} - 1}{D_{\alpha}} = 1 -
  D_{\alpha}^{-1}$. This completes the proof.
  \qed
\end{proof}

Theorem \ref{th:traceent} applies to all cluster states.  In
particular, for $n$ qubit cluster states $\ket{\text{C}_n}$ the
minimum vertex cover has the size $\abs{\beta} = \lfloor \frac{n}{2}
\rfloor$, yielding the cardinality $D_{\alpha} = 2^{\lfloor
  \frac{n}{2} \rfloor}$.  For even $n$, this yields $\Lam (
\ket{\text{C}_n} ) = 2^{-\frac{n}{2}}$ and $D_{\text{T}} (
\ket{\text{C}_n} , \delta ) = 1 - 2^{-\frac{n}{2}}$.

For graph states that do not satisfy the minimal rank condition, the
state \eqref{graph_css} generally does not minimize the three distance
measures, but it nevertheless yields upper bounds on the distances and
on the corresponding entanglement measures, the REE, the BE, and TE.

\section{Conclusion}\label{sec:con}

In this paper we reviewed and studied seven different definitions of
GM for arbitrary multipartite systems.  Five of these are known
($\Gc$, $\Gf$, $\Gcl$, $\Gm$, $\Gml$), one has previously received
only little interest ($\Gfl$), and one has not been studied before
($\Gt$).  The entanglement axioms of the measures were investigated,
and are summarized in Table \ref{GMproperties}.  A remaining open
question is whether $\Gcl$ satisfies weak monotonicity, something we
showed to be true at least for two-qubit states and isotropic states.
A complete quantitative hierarchy between the measures was derived
(shown in Figure \ref{GMhierarchy}), and it was found that this
hierarchy can be employed to partition the state space into pairwise
disjoint sets with clear physical properties (pure vs. mixed, and
separable vs. entangled). This is summarized in Table
\ref{tab:GMinequalities} and Figure \ref{statespace}.

As a byproduct of Corollary \ref{cor:traceent}, we found that for pure
input states $\rho = \proj{\psi}$ the trace distance $\DT ( \ket{\psi}
, \cdot )$ has in general no pure CSS.  This is in stark contrast to
the Bures distance, for which \eqref{purelambda2} implies that $\DB (
\ket{\psi} , \cdot )$ always has at least one pure CSS.  It is
therefore not trivial to find states for whom the Bures and trace
distance have a common CSS, something we did for a large class of
graph states in Theorem \ref{th:traceent}.

With regard to the convex roof-based measures $\Gc$ and $\Gcl$, it was
found that -- unlike $\Gf$ and $\Gfl$, or $\Gm$ and $\Gml$ -- these
two measures are not simple functions of each other, and in fact do
not even have the same ordering.  Nevertheless, some connections
between the two measures and their optimal decompositions could be
made (Lemma \ref{gcandgcl}, Theorem \ref{thm:Gfl=Gcl}, Corollary
\ref{coroptgcgcl}).  For this, the maximally correlated states were
particularly helpful, because their optimal decompositions for $\Gcl$
depend qualitatively on their parameters (Lemma \ref{le:Gcl>Gfl},
Proposition \ref{pp:twoqutritGcl}, Proposition \ref{pp:twoquditGcl}).
This way it could be shown in Corollary \ref{exclusiv} that for some
states $\Gc$ and $\Gcl$ do not share any common optimal decomposition.

For the linear GM it is known that the problem of finding the optimal
convex roof decomposition is equivalent to finding the closest
separable state for the fidelity ($\Gf \equiv \Gc$) \cite{skb11NJP}.
Somewhat surprisingly, we found that this is not the case for the
logarithmic GM. Already for bipartite systems the two problems are in
general inequivalent ($\Gfl \not\equiv \Gcl$) and need to be solved
separately.  Nevertheless, for two-qubit systems and for some classes
of states, such as all isotropic states and some maximally correlated
states, the two problems coincide.  While $\Gfl$ could be verified to
be weakly monotonous, the weak monotonicity of $\Gcl$ remains an open
problem.  At least for two-qubit systems this question can be answered
in the affirmative, because $\Gfl \equiv \Gcl$ then holds.  Another
open question is whether there exist states $\rho$ with $\Gcl( \rho )
= \Gfl ( \rho )$ for which the set of optimal decompositions of $\Gc (
\rho )$ is strictly larger than that of $\Gcl ( \rho )$, cf.~Theorem
\ref{thm:Gfl=Gcl}.

\begin{table}[htbp]
  \centering
  \caption{\label{GMproperties} Overview of the axioms fulfilled
    for the various definitions of GM. Subtable (a) lists the
    linear and logarithmic GM for pure states, and whether axioms
    are fulfilled when considering quantum operations between pure
    states only. Subtable (b)  lists the six distinct extensions
    of GM to mixed states. An axiom being fulfilled on pure states
    (as indicated in (a)) is necessary, but not sufficient
    for that axiom being fulfilled for mixed state extensions. The
    only exception is normalization, which is defined by pure states
    only. The properties of $\Gt$ and $\Gfl$ have not been studied
    before, and we found that $\Gcl$ satisfies Axiom 2(a) for
    two-qubit systems. For higher dimensions it is still unknown
    whether $\Gcl$ satisfies Axiom 2(a).}
  \subtable[Pure states $\ket{\Psi}$]{
    \begin{tabular}{c|cc}
      \hline
      \hline
      Properties & $\G$ & $\Gl$ \\
      \hline
      Axiom 1 & \yestick & \yestick \\
      Axiom 2(a) & \yestick & \yestick \\
      Axiom 2(b) & \yestick & \notick \\
      \hline
      Normalization & \notick & \yestick \\
      \hline
      \hline
    \end{tabular}
  }
  \hfill
  \subtable[Extensions to mixed states $\rho$]{
    \begin{tabular}{c|ccc|ccc}
      \hline
      \hline
      Properties &
      $\Gfc$ & $\Gm$ & $\Gt$ & $\Gfl$ & $\Gcl$ & $\Gml$ \\
      \hline
      Axiom 1 &
      \yestick & \notick & \notick & \yestick & \yestick & \notick \\
      Axiom 2(a) &
      \yestick & \notick & \notick & \yestick & ? & \notick \\
      Axiom 2(b) &
      \yestick & \notick & \notick & \notick & \notick & \notick \\
      \hline
      Convexity &
      \yestick & \notick & \notick & \yestick & \yestick & \notick \\
      Concavity &
      \notick & \yestick & \notick & \notick & \notick & \notick \\
      \hline
      \hline
    \end{tabular}
  }
\end{table}

As $\Gm$ and $\Gml$ assess the entanglement as well as the mixedness
of states \cite{hmm08}, they are not entanglement measures.  On the
other hand, $\Gfc$ is the only known definition of GM that yields a
strong entanglement measure. Because of its strong monotonicity and
convexity, $\Gfc$ never increases as states become more mixed.
Recalling that $\Gfc (\ket{\psi}) = \Gm (\ket{\psi})$ for pure states,
the inequality $\Gfc (\rho) \leq \Gm (\rho)$ becomes intuitively
clear.  What is more, Theorem \ref{thm:equivalence} states that this
inequality turns into an equality only if $\rho$ is pure.  Therefore,
$\Delta \G (\rho) := \Gm (\rho) - \Gfc (\rho)$ can be considered an
entropic quantity depending on the mixedness of $\rho$, akin to the
linear entropy. Equivalently, $\Delta \Gl (\rho) := \Gml (\rho) - \Gfl
(\rho)$ can be considered an alternative to the von Neumann entropy.
In contrast to this, $\Gml (\rho) - \Gcl (\rho)$ cannot be expected to
be a meaningful entropy quantifier, because this quantity can be zero
for genuinely mixed states (e.g.~maximally correlated states belonging
to set $D_3$, cf.~Theorem \ref{th:Dsubset}).  For such states the
relationship between their closest product state and optimal
decomposition was derived in Theorem \ref{thm:Gcl=Gml}.  It should
also be noted that $\Gm$ and $\Gml$ are readily accessible upper
bounds for all linear and logarithmic GM definitions, respectively.
These bounds can be easily computed from $\Lamm ( \rho )$, but they
become weaker as $\rho$ becomes more mixed.

The newly introduced extension of the linear GM by means of the trace
distance, $\Gt$, is not an entanglement measure, and it is yet unclear
what its benefits or operational implications for the study of
multipartite entanglement are. In contrast to this, the little-known
quantity $\Gfl$, which is closely related through $\Lamf$ to the
well-known definitions $\Gf$ and $\Gc$, has many desirable properties:
It satisfies normalization, convexity, weak monotonicity and is zero
for separable states. Most importantly, $\Gfl$ is so far the only
known definition of GM that yields a normalized entanglement measure
(for $\Gcl$ the question of weak monotonicity remains open).  We
therefore propose $\Gfl$ as the preferred definition of GM for studies
where normalization is a desirable feature, such as quantitative
entanglement characterization, entanglement scaling, or comparison
with other multipartite entanglement measures.

Even if $\Gcl$ should be verified as a weak entanglement measure,
$\Gfl$ has some other benefits over the logarithmic convex-roof: The
value of $\Gfl (\rho)$ is immediately known if $\Lamf (\rho)$ is
known, which is the case if either $\Gf(\rho)$ or $\Gc(\rho)$ have
been computed. Furthermore, $\Gfl$ is based on the fidelity, a widely
studied and physically meaningful distance in quantum information
theory, whereas $\Gcl$ is based on an abstract mathematical
definition.

In conclusion, for most situations either of the two fidelity-based
definitions $\Gfc$ or $\Gfl$ could be regarded as the best choice of
GM.  Since the two definitions have the same ordering and are closely
related via $\Lamf$, the choice depends merely on whether
normalization or strong monotonicity is more desirable.  The
entanglement of a state $\rho$ can then be found either by computing
the maximal fidelity to separable states, $\Lamf (\rho)$, or by
finding the optimal decomposition for the linear convex-roof.
Furthermore, the Bures entanglement \eqref{tr2} and the Groverian
entanglement $\text{E}_{\text{Gr}} ( \rho ) = \Gf ( \rho )^{1/2}$ are
closely related to $\Lamf$, and therefore to $\Gf$ and $\Gfl$
themselves, with all measures having the same ordering.  Knowing the
value of either of these provides lower bounds to $\Gm (\rho)$, $\Gcl
(\rho)$ as well as $\Gml (\rho)$. Finally, defining GM though the
maximal fidelity is also the most straightforward definition from a
historical viewpoint, because GM was originally defined by the maximal
fidelity between pure states.

\begin{acknowledgement}
  We thank Mark Wilde for pointing out an operational meaning of
  \eqref{LamFid}. LC was mainly supported by MITACS and NSERC. The CQT
  is funded by the Singapore MoE and the NRF as part of the Research
  Centres of Excellence programme.  MH acknowledges financial support
  from the Japan Society for the Promotion of Science (JSPS) through
  grant No.~23-01770, No.~23540463 and No.~23240001.
\end{acknowledgement}

\appendix

\section{Weak monotonicity of $\Gcl$ for two qubits}
\label{ap:weak}

It is known that Wootters' concurrence $C( \rho )$ is a weak
entanglement monotone for two-qubit states, i.e.~$C( \rho ) \geq C(
\sigma )$ holds for any trace-preserving quantum operation $\rho
\mapsto \sigma$ \cite{wootters98,wg03}.  Using the monotonically
increasing function $f(x) = -\log{\frac{1 + \sqrt{1-x^2}}{2}}$ and
Lemma \ref{le:GCL=twoqubit} below, it follows that $\Gcl ( \rho ) = f
( C ( \rho )) \geq f ( C ( \sigma )) = \Gcl ( \sigma )$.  So, $\Gcl$
is weakly monotonous for two qubit systems.
\begin{lemma}\label{le:GCL=twoqubit}
  Let $\rho$ be an arbitrary two-qubit state. Then, $\Gcl ( \rho ) =
  -\log{\frac{1 + \sqrt{1 - C( \rho )^2}}{2}}$, where $C( \rho )$ is
  the concurrence.
\end{lemma}
\begin{proof}
  The proof is similar to the derivation of $\Gc ( \rho )$ in
  \cite{wg03}.  First, it is easy to verify that the claim holds for
  pure states. The function $f(x) = -\log{\frac{1 + \sqrt{1-x^2}}{2}}$
  is monotonically increasing and convex for $x \in [0,1]$. Suppose
  $\rho = \sum_i P_i \proj{\Psi_i}$ is an optimal decomposition for
  $\Gcl ( \rho )$. Then
  \begin{equation}\label{ap12}
    \Gcl ( \rho ) =
    \sum_i P_i \Gl (\ket{\Psi_i}) =
    \sum_i P_i f(C(\ket{\Psi_i}))
    \geq f \Big( \sum_i P_i C(\ket{\Psi_i}) \Big) \geq
    f(C( \rho )) \, ,
  \end{equation}
  where the inequalities follow from the convexity of $f(x)$ and $C(
  \rho )$ \cite{wootters98}, respectively. On the other hand, Wootters
  found an optimal decomposition $\rho = \sum_i S_i \proj{\Phi_i}$ for
  the entanglement of formation, such that each $\ket{\Phi_i}$ has the
  same concurrence as $\rho$.  With this decomposition we obtain
  \begin{equation}\label{ap11}
    \Gcl ( \rho ) \leq
    \sum_{i} S_i \Gl ( \ket{\Phi_{i}} ) =
    \sum_{i} S_i f ( C ( \ket{\Phi_{i}} )) =
    f(C( \rho )) \, .
  \end{equation}
  From \eqref{ap11} and \eqref{ap12} it follows that $\Gcl ( \rho ) =
  f (C ( \rho ))$.
  \qed
\end{proof}

\section{Counterexample for concavity of $\Gt$}
\label{ap:gtpconcave}

Consider the isotropic state $\rho_{\iso} = p \frac{\one}{d^2} + (1-p)
\proj{\Psi}$, where $\ket{\Psi} = \frac{1}{\sqrt d} \sum^d_{i=1}
\ket{ii}$.  A counterexample for the concavity of $\Gt$ is found, if
\begin{equation}\label{ineqcond}
  \Gt ( \rho_{\iso} ) < p \Gt ( \one / d^2 ) + (1-p) \Gt ( \proj{\Psi} )
\end{equation}
holds for some $p \in (0,1)$ and $d \geq 2$.  Obviously, $\Gt (
\proj{\Psi} ) = \G ( \ket{\Psi} ) = 1 - \frac{1}{d}$, and because of
the isotropic nature of the maximally mixed state, for any $\ket{\phi}
\in \cH$ we can represent $\frac{\one}{d^2} - \proj{\phi}$ in matrix
form as $\text{diag} \big( \frac{1}{d^2} -1 , \frac{1}{d^2} , \ldots ,
\frac{1}{d^2} \big)$ by choosing a basis with $\ket{\phi}$ as the
first basis vector. Therefore,
\begin{equation*}
  \Gt \Big( \frac{\one}{d^2} \Big) =
  \! \min_{ \ket{\varphi} \in \pro} \! \frac{1}{4}
  \Big( \! \tr \Big\lvert \frac{\one}{d^2} -
  \proj{\varphi} \Big\rvert \Big)^2 =
  \frac{1}{4} \Big( \Big\lvert \frac{1}{d^2} - 1 \Big\rvert +
  (d^2 - 1) \frac{1}{d^2}  \Big)^2 = \Big( 1 - \frac{1}{d^2} \Big)^2 .
\end{equation*}
Now consider $\Gt ( \rho_{\iso} )$. From a geometric viewpoint the
relationship between $p \frac{\one}{d^2} + (1-p) \proj{\Psi}$ and
$\ket{\varphi} \in \pro$ is entirely determined by the angle between
$\ket{\Psi}$ and $\ket{\varphi}$. We therefore parameterize
$\ket{\varphi} = \sqrt{\alpha} \ket{\Psi} + \sqrt{1- \alpha}
\ket{\Psi^{\perp}}$, where $\ket{\Psi^{\perp}}$ is some state
orthogonal to $\ket{\Psi}$.  From $\braket{\Psi}{\varphi} =
\sqrt{\alpha}$, $\ket{\varphi} \in \pro$ and $\Lam ( \ket{\Psi} ) =
\frac{1}{d}$, it follows that $\alpha \leq \frac{1}{d}$.  To show that
this bound can be reached, we construct an example: The states $\{
\ket{\Psi_j} \}_{j=0}^{d-1}$ with $\ket{\Psi_j} = \frac{1}{\sqrt{d}}
\sum_{k=0}^{d-1} \E^{\I \frac{2 \pi jk}{d}} \ket{kk}$ form an
orthonormalized basis of MES. Obviously $\ket{\Psi} = \ket{\Psi_0}$,
and defining $\ket{\Psi^{\perp}} := \frac{1}{\sqrt{d-1}}
\sum_{i=1}^{d-1} \ket{\Psi_i}$, we obtain $\ket{\varphi} = \sqrt{1/d}
\ket{\Psi} + \sqrt{(d-1)/d} \ket{\Psi^{\perp}} = \frac{1}{d}
\sum_{j=0}^{d-1} \sum_{k=0}^{d-1} \E^{\I \frac{2 \pi jk}{d}} \ket{kk}
= \ket{00} \in \pro$.

We now write
\begin{multline}\label{isoform}
  \rho_{\iso} - \proj{\varphi} =
  p \frac{\one}{d^2} + (1-p) \proj{\Psi} - \alpha \proj{\Psi} \\ -
  \sqrt{\alpha (1- \alpha )} \Big( \ketbra{\Psi}{\Psi^{\perp}}
  + \ketbra{\Psi^{\perp}}{\Psi} \Big)
  - (1- \alpha ) \proj{\Psi^{\perp}} \, ,
\end{multline}
and using $\ket{\Psi}$ and $\ket{\Psi^{\perp}}$ as the first two basis
vectors for the matrix representation of \eqref{isoform}, we obtain
\begin{equation*}
  \rho_{\iso} - \proj{\varphi} =
  \begin{pmatrix}
    \frac{p}{d^2} + 1-p - \alpha &
    - \sqrt{ \alpha ( 1 - \alpha) } & 0 & \cdots & 0 \\
    - \sqrt{ \alpha ( 1 - \alpha) } &
    \frac{p}{d^2} - 1 + \alpha & 0 & \cdots & 0 \\
    0 & 0 & \frac{p}{d^2} & \cdots & 0 \\
    \vdots & \vdots & \vdots & & \vdots \\
    0 & 0 & 0 & \cdots & \frac{p}{d^2}
  \end{pmatrix}
  \, .
\end{equation*}
The eigenvaules of this matrix are $\{ \frac{p}{d^2} - \frac{p}{2} \pm
\frac{1}{2} \sqrt{(2-p)^2 - 4 \alpha (1-p) } , \frac{p}{d^2} , \ldots
, \frac{p}{d^2} \}$.  The radicand in the first two eigenvalues is
positive for $p \in (0,1)$ and $\alpha \in [0, \frac{1}{d} ]$. The
first and second eigenvalue are positive and negative, respectively,
for all $p \in (0,1)$, $\alpha \in [0 , \frac{1}{d} ]$ and $d \geq
2$. Hence,
\begin{multline*}
  \Gt ( \rho_{\iso} ) =
  \min_{\ket{\varphi} \in \pro} \frac{1}{4}
  \Big( \tr \Big\lvert p \frac{\one}{d^2} +
  (1-p) \proj{\Psi} - \proj{\varphi} \Big\rvert \Big)^2 \\
  = \min_{\alpha \in [0, \frac{1}{d}]}
  \frac{1}{4} \Big[ \sqrt{(2-p)^2 - 4 \alpha (1-p) } +
  (d^2 - 2) \frac{p}{d^2} \Big]^2 \, ,
\end{multline*}
and the minimum is obviously reached for $\alpha = \frac{1}{d}$.  We
can now rewrite \eqref{ineqcond} as
\begin{equation*}
  \frac{1}{4} \bigg[ \sqrt{(2-p)^2 - \frac{4}{d} (1-p) } +
  (d^2 - 2) \frac{p}{d^2} \bigg]^2 <
  p \Big( 1 - \frac{1}{d^2} \Big)^2 +
  (1-p) \Big( 1 - \frac{1}{d} \Big) \, ,
\end{equation*}
and it can be easily verified numerically that this inequality is
satisfied for all $p \in ( 0 , 1)$ and all $d \geq 2$.  Hence, $\Gt$
is not concave.

\section{Auxiliary results for calculation of $\Gcl$ for maximally correlated states}\label{ap:minimum}

\begin{lemma}\label{le:inequality}
  Let $k \in \mathbb{N}$, $n > 0$, $X > 0$ and $Y \geq 0$ be
  constants, and let $x_1, \ldots , x_k > 0$, $y_1, \ldots , y_k \geq
  0$ be variables with the restrictions $\sum^{k}_{i=1} x_i = X$ and
  $\sum^{k}_{i=1} y_i = Y$.  Then
  \begin{equation}\label{ea:inequality}
    \min_{\substack{ \{ x_1 , \ldots , x_k \} \\
        \{ y_1 , \ldots , y_k \} }}
    \sum^{k}_{i=1} (x_i+y_i) \log \bigg[ n \Big( 1 +
    \frac{y_i}{x_i} \Big) \bigg]  =
    ( X + Y )
    \log \bigg[ n \Big( 1 + \frac{Y}{X} \Big) \bigg] \, .
  \end{equation}
\end{lemma}
\begin{proof}
  Because of $\log [ n (1 + \frac{y_i}{x_i}) ] = \log n + \log (1 +
  \frac{y_i}{x_i})$, it suffices to consider $n=1$.
  Eq.~\eqref{ea:inequality} clearly holds for $k=1$, so we assume $k
  \geq 2$.  First, we consider the $x_i$ to be fixed, with their
  values denoted as $\{ x_1' , \ldots , x_k' \}$.  Then, the function
  \begin{equation}\label{fform1}
    f(y_1, \ldots , y_k ) :=
    \sum^{k}_{i=1} (x_i' + y_i) \log \Big(1+ \frac{y_i}{x_i'} \Big)
  \end{equation}
  is a function of $k$ variables.  We use this function to define the
  $k$ functions of $k-1$ variables,
  \begin{equation}\label{fform2}
    g_{j} (y_1 , \ldots , y_{j-1} , y_{j+1} , \ldots , y_k ) :=
    f(y_1 , \ldots , y_{j-1}, Y_j , y_{j+1}, \ldots , y_k ) \, ,
  \end{equation}
  with $j \in [1,k]$, and where $Y_j$ is short-hand for $Y_j := Y -
  \sum_{i \neq j} y_i$.  Obviously, if \eqref{fform1} is minimized at
  $( y_1' , \ldots , y_k' )$ under the condition $\sum^{k}_{i=1} y_i =
  Y$, then \eqref{fform2} is minimized at $( y_1' , \ldots , y_{j-1}'
  , y_{j+1}' ,\ldots , y_k' )$.  In particular, all partial
  derivatives must vanish at this point, and from \eqref{fform1} and
  \eqref{fform2} it follows that
  \begin{equation}\label{fform3}
    \frac{\partial g_{j}}{\partial y_l} =
    \log \frac{x_j' x_l' + x_j' y_l}{x_j' x_l' + Y_j x_l'} \, , \quad
    \text{for all }
    j, l \in [1,k] \text{ with } j \neq l \, .
  \end{equation}
  To rule out boundary points as solutions, we calculate the second
  derivatives,
  \begin{equation*}
    \frac{\partial^2 g_{j}}{\partial y^2_l} =
    \frac{x_l' + y_l + x_j' + Y_j}{(x_l' + y_l ) ( x_j' + Y_j )} > 0 \, ,
  \end{equation*}
  and find that they are strictly positive. From this it follows that
  \eqref{fform3} can be zero at only one point, and that the unique
  minimum is reached there. Since $x_j' y_l = Y_j x_l'$ results in
  $\frac{\partial g_{j}}{\partial y_l} = 0$, we obtain $x_j' y_l' =
  y_j' x_l'$ for all $j, l \in [1,k]$ with $j \neq l$.  From this it
  follows that $x_j' (\sum_i y_i') = y_j' (\sum_i x_i')$, and thus
  $\frac{y_j'}{x_j'} = \frac{Y}{X}$ for all $j \in [1,k]$. Inserting
  this in the l.h.s.~of \eqref{ea:inequality} yields the r.h.s.~of
  \eqref{ea:inequality}, regardless of the initial choice of the fixed
  $\{ x_1' , \ldots , x_k' \}$.  This completes the proof.
  \qed
\end{proof}

\begin{lemma}\label{le:fhs}
  For constants $m, n \in \mathbb{N}$ with $\frac{m}{n} \leq 1$ and $q
  \in (0,1)$, the minimum of the function
  \begin{equation*}
    f (h,s) = \frac{q - s(1-q)}{1 - hs} (1 + h)
    \log \big( m ( 1 + h ) \big) + \frac{(1-q) - hq}{1 - hs}
    (1 + s) \log \big( n ( 1 + s) \big)
  \end{equation*}
  with the domain given by $h\in [0,\frac{1-q}{q}]$ and $s\in
  [0,\frac{q}{1-q}]$, excluding the point $( \frac{1-q}{q},
  \frac{q}{1-q})$, is
  \begin{equation*}
    \min f (h,s) =
    \left\{
      \begin{array}{l}
        q \log m + (1-q) \log n \hfill \text{for }
        \tfrac{m}{n} \geq \frac{1}{\E} \\
        \log ( \tfrac{m}{q} ) \hfill \text{for }
        \tfrac{m}{n} < \tfrac{1}{\E} \text{ and } q \geq \tfrac{\E m}{n} \\
        \log n - q \tfrac{n \log \E}{m \E} \qquad \qquad \quad
        \text{for }
        \tfrac{m}{n} < \tfrac{1}{\E} \text{ and } q < \tfrac{\E m}{n}
      \end{array} \right.
  \end{equation*}
\end{lemma}
\begin{proof}
  The function $f (h,s)$ is continuous and differentiable in its
  entire domain, with the partial derivatives
  \begin{align*}
    \frac{\partial f (h,s)}{\partial h} &=
    \frac{q - (1-q) s}{(1 - hs)^2 \ln 2}
    \bigg[ 1 - hs + (1+s) \ln
    \Big( \frac{m (1+h)}{n (1+s)} \Big) \bigg] \, , \\
    \frac{\partial f(h,s)}{\partial s} &=
    \frac{(1-q) - hq}{(1 - hs)^2 \ln 2}
    \bigg[ 1 - hs + (1+h) \ln \Big( \frac{n (1+s)}{m (1+h)} \Big)
    \bigg] \, .
  \end{align*}
  For $f (h,s)$ to attain a minimum at an interior point, both partial
  derivatives must vanish, i.e.
  \begin{equation*}
    1 - hs + (1+s) \ln \Big( \frac{m (1+h)}{n (1+s)} \Big) =
    1 - hs + (1+h) \ln \Big( \frac{n (1+s)}{m (1+h)} \Big) = 0 \, ,
  \end{equation*}
  from which it follows that $\frac{1+s}{1+h} = -1$. Since this
  cannot be true for any $h$ and $s$, the minimum of $f(h,s)$ must be
  attained at a boundary point. The boundaries of $f (h,s)$ are
    \begin{align*}
      f ( h, \tfrac{q}{1-q} ) &= \log ( \tfrac{n}{1-q} )
      &\text{for } h \in [0, \tfrac{1-q}{q}) \\
      f ( \tfrac{1-q}{q} , s ) &= \log ( \tfrac{m}{q} )
      &\text{for } s \in [0, \tfrac{q}{1-q}) \\
      f_h (h) := f (h,0) &=
      \log n + q (1 + h) \log \big( \tfrac{m ( 1 + h )}{n} \big)
      &\text{for } h \in [0, \tfrac{1-q}{q}] \\
      f_s (s) := f (0,s) &=
      \log m + (1-q) (1 + s) \log \big( \tfrac{n ( 1 + s)}{m} \big)
      &\text{for } s \in [0, \tfrac{q}{1-q}]
    \end{align*}
    Because $f$ is constant on the first two boundaries, and because
    $f_h (\frac{1-q}{q}) = \log (\frac{m}{q})$ and $f_s
    (\frac{q}{1-q}) = \log (\frac{n}{1-q})$, it suffices to find the
    minimum of $f_h$ and $f_s$.  From $\frac{n}{m} \geq 1$ it is clear
    that $f_s (s)$ is monotonically increasing in $s$, hence $\min_s
    f_s (s) = f_s (0) = f(0,0) = q \log m + (1-q) \log n \geq \min_h
    f(h,0) = \min_h f_h (h)$.  Therefore, it suffices to find the
    minimum of $f_h$.  The derivative of $f_h$ is
  \begin{equation*}
    f'_h (h) = \frac{\partial f (h,0)}{\partial h} =
    q \log \frac{\E m (1+h)}{n} \, .
  \end{equation*}
  For $\frac{m}{n} \geq \frac{1}{\E}$, we have $f'_h (h) \geq 0$,
  hence $f_h$ attains its minimum at $f_h (0) = f (0,0)$. For
  $\frac{m}{n} < \frac{1}{\E}$ and $q \geq \frac{\E m}{n}$, we have
  $f'_h (h) \leq 0$, hence the minimum is $f_h (\frac{1-q}{q}) = \log
  (\frac{m}{q})$.  On the other hand, for $\frac{m}{n} < \frac{1}{\E}$
  and $q < \frac{\E m}{n}$, then $f'_h ( \frac{n}{\E m} - 1 ) = 0$, so
  the minimum of $f_h$ is $f_h ( \frac{n}{\E m} - 1 ) = \log n - q
  \frac{n \log \E}{m \E}$.
  \qed
\end{proof}

\end{document}